\documentclass{article} %
\usepackage{iclr2026_conference,times}

\usepackage{amsmath,amsfonts,bm}

\def\eqref#1{equation~\ref{#1}}

\def\1{\bm{1}}

\DeclareMathAlphabet{\mathsfit}{\encodingdefault}{\sfdefault}{m}{sl}
\SetMathAlphabet{\mathsfit}{bold}{\encodingdefault}{\sfdefault}{bx}{n}

\DeclareMathOperator{\sign}{sign}

\usepackage{xcolor}         %
\usepackage{xspace}         %
\usepackage[utf8]{inputenc} %
\usepackage[T1]{fontenc}    %
\usepackage{hyperref}       %
\usepackage{url}            %
\usepackage{booktabs}       %
\usepackage{amsfonts}       %
\usepackage{nicefrac}       %
\usepackage{microtype}      %
\usepackage{xcolor}
\usepackage{amsmath}
\usepackage{multirow}
\usepackage{array}
\usepackage{siunitx}
\usepackage{booktabs}
\usepackage{comment}
\usepackage{tikz}
\usepackage{listings}
\usepackage{wrapfig}
\usepackage{cleveref}
\usepackage{svg}

\usepackage{packages/default_packages}
\usepackage{packages/colors}
\usepackage{packages/colab}
\usepackage{packages/mjb}
\usepackage{packages/local_defs}

\newcommand\extrafootertext[1]{%
    \bgroup
    \renewcommand\thefootnote{\fnsymbol{footnote}}%
    \renewcommand\thempfootnote{\fnsymbol{mpfootnote}}%
    \footnotetext[0]{#1}%
    \egroup
}

\newcommand{\datasetname}{LIMIT}

\graphicspath{{figures/}}

\title{Theoretical Limits of Embeddings in the Age of Reasoning and Instruction-Following Retrieval}

\iclrfinalcopy
\author{
    \textbf{Orion Weller\textsuperscript{1,2}}
    \quad
    \textbf{Michael Boratko\textsuperscript{1}}
    \quad
    \textbf{Iftekhar Naim\textsuperscript{1}}
    \quad
    \textbf{Jinhyuk Lee\textsuperscript{1}}  \\ \\
    \textsuperscript{1}Google DeepMind, \textsuperscript{2}Johns Hopkins University \\ \\
    \texttt{oweller@cs.jhu.edu,jinhyuklee@google.com}
}

\begin{document}

\newcommand{\fix}{\marginpar{FIX}}
\newcommand{\new}{\marginpar{NEW}}

\title{On the Theoretical Limitations of \\Embedding-Based Retrieval}

\maketitle

\begin{abstract}
Vector embeddings have been tasked with an ever-increasing set of retrieval tasks over the years, with a nascent rise in using them for reasoning, instruction-following, coding, and more. These new benchmarks push embeddings to work for \emph{any query} and \emph{any notion of relevance} that could be given. While prior works have pointed out theoretical limitations of vector embeddings, there is a common assumption that these difficulties are exclusively due to unrealistic queries, and those that are not can be overcome with better training data and larger models.
In this work, we demonstrate that we may encounter these theoretical limitations in realistic settings with extremely simple queries.
We connect known results in learning theory, showing that the number of top-$k$ subsets of documents capable of being returned as the result of some query is limited by the dimension of the embedding.
We empirically show that this holds true even if we directly optimize on the test set with free parameterized embeddings.
We then create a realistic dataset called \datasetname{} that stress tests embedding models based on these theoretical results, and observe that even state-of-the-art models fail on this dataset despite the simple nature of the task.
Our work shows the limits of embedding models under the existing single vector paradigm and calls for future research to develop new techniques that can resolve this fundamental limitation.
\end{abstract}

\newcommand{\draftonly}[1]{#1}
\newcommand{\eat}[1]{}
\renewcommand{\draftonly}[1]{}
\definecolor{darkgreen}{RGB}{0, 102, 0}

\newcommand{\draftcomment}[3]{\draftonly{\textcolor{#2}{{{[#1: #3]}}}}}

\newcommand{\orion}[1]{\textcolor{brown}{[OW: #1]}}
\newcommand{\jinhyuk}[1]{\textcolor{blue}{[JL: #1]}}
\newcommand{\mjb}[1]{\textcolor{violet}{[MB: #1]}}
\newcommand{\inaim}[1]{\textcolor{teal}{[IN: #1]}}

\crefformat{section}{\S#2#1#3} %
\crefformat{subsection}{\S#2#1#3}
\crefformat{subsubsection}{\S#2#1#3}

\section{Introduction}

Over the last two decades, information retrieval (IR) has moved from models dominated by sparse techniques (such as BM25 \citet{robertson1995okapi}) to those that use neural language models (LM) as their backbones \citep{lee-etal-2019-latent,craswell2020overview,izacard2021unsupervised,wang2022text}.
These neural models are predominantly used in a single vector capacity, where they output a single \textit{embedding} representing the entire input (also known as \textit{dense retrieval}).
These embedding models are capable of generalizing to new retrieval datasets and have been tasked with solving increasingly complicated retrieval problems \citep{thakur2021beir,enevoldsen2025mmteb,lee2025gemini}.

In recent years this has been pushed even further with the rise of instruction-following retrieval benchmarks, where models are asked to represent \textbf{any relevance definition} for \textbf{any query} \citep{weller2025mfollowir,weller2025rank1,song2025ifir,xiao2024rar,su2024bright}.
For example, the QUEST dataset \citep{malaviya2023quest} uses logical operators to combine different concepts, studying the difficulty of retrieval for complex queries (e.g., "Moths or Insects or Arthropods of Guadeloupe"). On the other hand, datasets such as BRIGHT \citep{su2024bright} explore the challenges arising from different definitions of relevance by defining relevance in ways that require reasoning. One subtask includes reasoning over a given Leetcode problem (the query) to find other Leetcode problems that share a subtask (e.g. others problems using dynamic programming).
Although models cannot solve these benchmarks yet, the community has proposed these problems in order to push the boundaries of what dense retrievers are capable of---which is now implicitly \textit{every task} that could be defined.

Rather than proposing empirical benchmarks to gauge what embedding models can achieve, we seek to understand at a more fundamental level what the limitations are.
Since embedding models use vector representations in geometric space, there exist well-studied fields of mathematical research \citep{papadimitriou1982communication} that could be used to analyze these representations.

Our work aims to bridge this gap, connecting known theoretical results in linear algebra with modern advancements in neural information retrieval.
We draw upon research in high-dimensional geometry to provide a lower bound on the embedding dimension needed to represent a given combination of relevant documents and queries.
Specifically, we show that for a given embedding dimension $d$ \textbf{there exist top-$k$ combinations of documents that cannot be returned}---no matter the query---highlighting a theoretical and fundamental limit to embedding models.

 \begin{figure*}[t!]
    \vspace{-1em}
    \includegraphics[width=0.99\linewidth,trim=0cm 0.5cm 0cm 0cm]{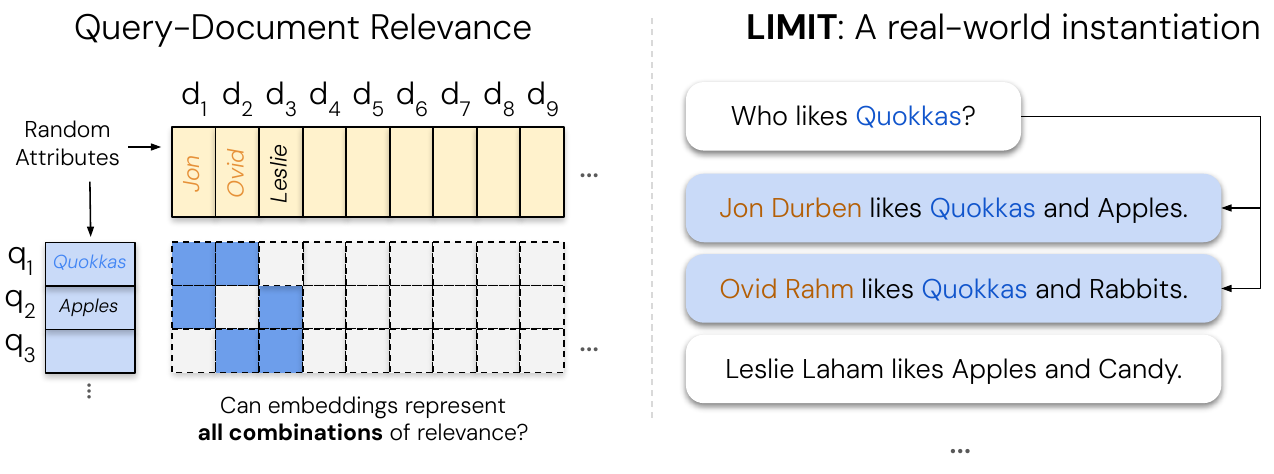}
    \captionsetup{justification=raggedright, singlelinecheck=false}
    \vspace{0.2em}
    \caption{A depiction of the \datasetname{} dataset creation process, based on theoretical limitations. We test \textbf{all combinations} of relevance for $N$ documents (i.e. in the figure, all combinations of relevance for three documents with two relevant documents per query) and instantiate it using a simple mapping. 
    }
    \label{fig:teaser}
    \vspace{-2em}
\end{figure*}

To show that this theoretical limit is true for \text{any} retrieval model or training dataset, we test a setting where the vectors themselves are directly optimized with the test data. This allows us to empirically show how the embedding dimension enables the solving of retrieval tasks. We find that there exists a crucial point for each embedding dimension ($d$) where the number of documents is too large for the embedding dimension to encode all combinations. We then gather these crucial points for a variety of $d$ and show that this relationship can be modeled empirically with a polynomial function.

We also go one step further and construct a realistic but simple dataset based on these theoretical limitations (called \datasetname{}).\footnote{Data and code are available at \url{https://github.com/google-deepmind/limit}
} Despite the simplicity of the task (e.g., \texttt{who likes Apples?} and \texttt{Jon likes Apples, ...}), we find it is very difficult for even state-of-the-art embedding models~\citep{lee2025gemini,qwen3embedding} on MTEB~\citep{enevoldsen2025mmteb}, and practically impossible for models with small embedding dimensions using standard optimization techniques.  

Overall, our work contributes: (1) a theoretical basis for the fundamental limitations of embedding models, (2) a best-case empirical analysis showing that this proof holds for any dataset instantiation (by free embedding optimization), and (3) a simple real-world natural language instantiation called \datasetname\ that even state-of-the-art embedding models cannot solve.

These results imply interesting findings for the community: on one hand we see neural embedding models becoming immensely successful. However, academic benchmarks test only a small amount of the queries that could be issued (and these queries are often overfitted to), hiding these limitations.
Our work shows that as the tasks given to embedding models require returning ever-increasing combinations of top-$k$ relevant documents (e.g., through instructions connecting previously unrelated documents with logical operators), we will reach a limit of combinations they can represent. 

Thus, the community should be aware of these limitations, both when creating evals and also by using alternate architectures---such as cross-encoders / multi-vector / more expressive similarity functions ---when trying to handle the full range of instruction queries, i.e. \textit{any query and relevance definition}.

\section{Related Work}
\subsection{Neural Embedding Models}
There has been immense progress on embedding models in recent years \citep{lee-etal-2019-latent,craswell2020overview,behnamghader2024llm2vec}, moving from simple web search (text-only) to advanced instruction-following and multi-modal representations.  These models generally followed advancements in language models, such as pre-trained LMs \citep{hoffmann2022training}, multi-modal LMs \citep{li2024multimodal,team2024chameleon}, and advancements in instruction-following \citep{zhou2023instruction,ouyang2022training}. Some of the prominent examples in retrieval include CoPali \citep{faysse2024colpali} and DSE \citep{ma2024unifying} which focus on multimodal embeddings, Instructor \citep{su2022one} and FollowIR \citep{weller2024followir} for instruction following, and GritLM \citep{muennighoff2024generative} and Gemini Embeddings \citep{lee2025gemini} for pre-trained LMs turned embedders. 

Our work, though focused solely on textual representations for simplicity, \textbf{applies to all modalities of single vector embeddings for any domain of dataset}. As the space of things to represent grows (through instructions or multi-modality) they will increasingly run into these theoretical limitations.

\subsection{Empirical tasks pushing the limits of dense retrieval}
Retrieval models have been pushed beyond their initial use cases to handle a broad variety of areas. Notable works include efforts to represent a wide group of domains \citep{thakur2021beir,lee2024gecko}, a diverse set of instructions \citep{weller2024followir,Zhou2024BeyondCR,oh2024instructir,weller2025seq}, and to handle reasoning over the queries \citep{xiao2024rar,su2024bright}. This has pushed the focus of embedding models from basic keyword matching to embeddings that can represent the full semantic meaning of language. As such, it is more common than ever to connect what were previously unrelated documents into the top-$k$ relevant set,\footnote{You can imagine an easy way to connect any two documents merely by using logical operators, i.e. X and Y.} increasing the number of combinations that models must be able to represent. This has motivated our interest in understanding the limits of what embeddings can represent, as current work expects it to handle \textit{every} task. 

Previous work has explored empirically the limits of models: \citet{reimers2020curse} showed that smaller dimension embedding models have more false positives, especially with larger-scale corpora. \citet{ormazabal2019analyzing} showed the empirical limitations of models in the cross-lingual setting and \citet{yin2018dimensionality} showed how embedding dimensions relate to the bias-variance tradeoff. In contrast, our work provides a theoretical connection between the embedding dimension and the top-k sets it can retrieve, while also showing empirical limitations.

\subsection{Theoretical Limits of Vectors in Geometric Space}
\label{sec:related_math}
Understanding and finding nearest neighbors in semantic space has a long history in mathematics research, with early work such as the Voronoi diagram being studied as far back as 1644 and formalized in 1908 \citep{voronoi1908nouvelles}. The order-$k$ version of the Voronoi diagram (i.e. the Voronoi diagram partitioning the space into regions based on their closest $k$ points) is obviously connected to information retrieval and has been studied for many years \citep{clarkson1988applications}. The number of such regions is equal to the number of unique retrieval sets of size $k$, however this quantity is notoriously difficult to bound tightly
\citep{BOHLER2015539,lee1982k,chen2023efficient}. 

We approach this problem from a different angle, asking not how many $k$-subsets a given configuration realizes, but rather what embedding dimension is \emph{necessary} to realize all $k$-subsets with a guaranteed score margin. By applying a classical sphere-packing volume argument \citep{vershynin2018high,conway1999recent}, we obtain a lower bound on the embedding dimension in terms of $n$, $k$, and the margin~$\gamma$. Our result is conceptually related to the Johnson--Lindenstrauss lemma \citep{johnson1984extensions}, which gives a \emph{sufficient} dimension to preserve pairwise distances among $n$ points; in contrast, our bound gives a \emph{necessary} dimension to realize all retrieval sets with a margin. The role of the margin in controlling the complexity of realizable configurations also parallels classical results in statistical learning theory, including the fat-shattering dimension \citep{kearns1994efficient} and margin-based generalization bounds for linear classifiers \citep{bartlett1998sample,vapnik1998statistical}, where larger margins similarly constrain the capacity of the hypothesis class.

\section{Representational Capacity of Vector Embeddings}
In this section we formally define the minimum embedding dimension required to satisfy a given retrieval objective, and draw on classical sphere-packing results from high-dimensional geometry to establish a lower bound. We note that this will be an extreme lower bound, as practical models have to deal with other constraints such as learning through gradient descent and using LM tokenization.

\paragraph{Setup.}
Let $v_1,\dots,v_n\in\mathbb{R}^d$ be unit\footnote{For simplicity, as nearly all SoTA retrieval models use unit vectors.} document vectors, and let queries be unit vectors $u\in\mathbb{R}^d$. Fix $\gamma>0$. A $k$-subset $S\subseteq[n]$ is realized with margin $\gamma$ if there exists a unit query $u_S$ such that

\begin{equation}
\min_{i\in S}\langle u_S,v_i\rangle\ \ge\ \max_{j\notin S}\langle u_S,v_j\rangle\ +\ 2\gamma.
\label{eq:margin}
\end{equation}

Since $\langle u,v_i\rangle\in[-1,1]$ for unit vectors, any score gap is at most $2$, hence \eqref{eq:margin} is feasible only for $0<\gamma\le 1$. Throughout, $\log$ denotes the natural logarithm.

\begin{theorem}[Dimension lower bound]
\label{thm:margin-lb}
Assume $1\le k<n$ and that \emph{every} $k$-subset $S\subseteq[n]$ is realized with margin $\gamma$ as in \eqref{eq:margin}. Then

\begin{equation}
\binom{n}{k}\ \le\ \Bigl(1+\frac{1}{\gamma}\Bigr)^{d},
\qquad\text{hence}\qquad
d\ \ge\ \frac{\log\binom{n}{k}}{\log\!\bigl(1+1/\gamma\bigr)}.
\label{eq:bound}
\end{equation}
\end{theorem}

\begin{proof}
Fix two distinct $k$-subsets $S\neq T$ and choose $i\in S\setminus T$ and $j\in T\setminus S$. Applying \eqref{eq:margin} to $S$ and to $T$ gives

\[
\langle u_S, v_i-v_j\rangle\ge 2\gamma,
\qquad
\langle u_T, v_j-v_i\rangle\ge 2\gamma.
\]

Adding yields $\langle u_S-u_T, v_i-v_j\rangle\ge 4\gamma$. 
By Cauchy--Schwarz and (universally, for \emph{any} unit vectors) $\|v_i-v_j\|\le \|v_i\|+\|v_j\|=2$, we obtain
$\|u_S-u_T\|\ge 2\gamma$.
Thus the $M=\binom{n}{k}$ unit queries $\{u_S\}$ are pairwise $2\gamma$-separated, so the open balls $B(u_S,\gamma)$ are disjoint.
Moreover, since $\|u_S\|=1$, each $B(u_S,\gamma)\subseteq B(0,1+\gamma)$, and therefore
\[
M\cdot \mathrm{vol}\!\bigl(B_d(\gamma)\bigr)\ \le\ \mathrm{vol}\!\bigl(B_d(1+\gamma)\bigr).
\]
Using $\mathrm{vol}(B_d(r))=C_d\,r^d$ for a constant $C_d$ depending only on $d$ (which cancels), we get $M\gamma^d\le (1+\gamma)^d$, i.e. $\binom{n}{k}\le \bigl((1+\gamma)/\gamma\bigr)^d=\bigl(1+1/\gamma\bigr)^d$.
Rearranging yields \eqref{eq:bound}.
\end{proof}

\begin{table}[h!]
\centering
\footnotesize
\vspace{-1em}
\begin{tabular}{l|cccc}
\toprule
Corpus size $n$ & $k=2$ & $k=10$ & $k=100$ & $k=1000$ \\
\midrule
$10^2$  & 4  & 13  & \textit{trivial} & ---\\
$10^3$  & 6  & 23  & 135 & \textit{trivial}\\
$10^4$  & 8  & 33  & 233 & 1354\\
$10^5$  & 10 & 42  & 329 & 2334\\
$10^6$  & 12 & 52  & 425 & 3296\\
$10^7$  & 14 & 61  & 521 & 4257\\
$10^8$  & 16 & 71  & 617 & 5217\\
$10^9$  & 17 & 81  & 713 & 6177\\
$10^{10}$ & 19 & 90  & 809 & 7137\\
$10^{11}$ & 21 & 100 & 905 & 8098\\
\bottomrule
\end{tabular}
\captionsetup{justification=raggedright, singlelinecheck=false}
\caption{Lower bounds for embedding dimension from Theorem~\ref{thm:margin-lb} for $\gamma=0.1$. When $n$ and $k$ are both 1000, the result is \textit{trivial} because $\binom{1000}{1000}=1$ (there is only one $k$-subset, hence no ``irrelevant'' items to separate). We see for large $k$ and $n$ values these dimension requirements are already greater than those currently used for web-scale search. \textbf{If these numbers are inflated by a small multiple due to constraints on gradient learning or other LM-based constraints} (e.g. tokenization, generalization) \textbf{these bounds are outside of any reasonable embedding dimension.}}
\label{tab:bounds}
\end{table}

\subsubsection{Implications}

\paragraph{Numerical instantiation}
We can illustrate the effects of this lower bound using $\gamma=0.1$ (score gap $2\gamma=0.2$), which is approximately standard for models based on empirical usage. Thus, \eqref{eq:bound} becomes
$d\ge \left\lceil \log\binom{n}{k}/\log 11\right\rceil$, with a table for various $k$ and $n$ values in Table~\ref{tab:bounds}.

For $n\gg k$, $\log\binom{n}{k}\approx k\log(en/k)$, so \eqref{eq:bound} forces
\[
d\ =\ \Omega\!\left(\frac{k\log(en/k)}{\log(1+1/\gamma)}\right).
\]
A stricter margin requirement (larger $\gamma$) demands higher dimension, since $\log(1+1/\gamma)$ decreases with $\gamma$ (feasibility requires $\gamma\le 1$, so the denominator is at least $\log 2$). 

\paragraph{Consequences}
Due to space and speed requirements, most embeddings used for web-scale search are quantized or truncated (e.g. through Matryoshka embeddings \citep{kusupati2022matryoshka}) to less than 1k dimensions, while the largest embeddings used in research are around 4096 \citep{qwen3embedding}. We see that even with a moderate margin, which is needed to handle noise from messy data or quantization, the lower bounds in Table~\ref{tab:bounds} can already be larger than what is used in practice.

Additional constraints on real-world models (such as needing to generalize, learn from gradient descent, and use natural language and tokenization) will make the dimension required in practice much higher. As Table~\ref{tab:bounds} shows, even a small multiple of this lower bound would make the embedding dimension requirement infeasible. This multiple seems well-founded, as we will show in the next section from the best-case optimization setting (e.g. free embeddings).

\section{Empirical Connection: Best Case Optimization}
\label{sec:free_embeds}
Having established a theoretical limitation of embedding models based on their embedding dimension $d$, we seek to show that this holds empirically also.

To show the strongest optimization case possible, we design experiments where the vectors themselves are directly optimizable with gradient descent.\footnote{This could also be viewed as an embedding model where each query/doc are a separate vector via lookup.} We call this "free embedding" optimization, as the embeddings are free to be optimized and not constrained by natural language, which imposes constraints on any realistic embedding model. Thus, this shows whether it is feasible for \textbf{any embedding model} to solve this problem: if the free embedding optimization cannot solve the problem, real retrieval models will not be able to either. It is also worth noting that we do this by directly optimizing the embeddings over the target qrel matrix (test set). This will not generalize to a new dataset, but is done to show the highest performance that could possibly occur.

\paragraph{Experimental Settings}
We create a random document matrix (size $n$) and a random query matrix with top-$k$ sets (of all combinations, i.e. size $m=\binom{n}{k}$), both with unit vectors. We then directly optimize for solving the constraints with the Adam optimizer \citep{kingma2014adam}.\footnote{We found similar results with SGD, but we use Adam for speed and similarity with existing training methods.}
Each gradient update is a full pass through all correct triples (i.e. full dataset batch-size) with the InfoNCE loss function \citep{oord2018representation},\footnote{In preliminary experiments, we found that InfoNCE performed best, beating MSE and Margin. As we are directly optimizing the vectors with full-dataset batches, this is $\mathcal{L}_{\text{total}} = -\frac{1}{M} \sum_{i=1}^{M} \log \frac{\sum_{d_r \in R_i} \exp(\text{sim}(q_i, d_r) / \tau)}{\sum_{d_k \in D} \exp(\text{sim}(q_i, d_k) / \tau)}$ where $D$ is all docs, $d_r$ is the relevant documents for query $q_i$ and $d_k$ are the non-relevant documents. For experiments with sigmoid learning functions (e.g. \citet{bangachev2025global}, see Appendix~\ref{app:sigmoid})} with all other documents as in-batch negatives (i.e. full dataset in batch). As nearly all embedding models use normalized vectors, we do also (via projected gradient descent).
We perform early stopping when there is no improvement in the loss for 1000 iterations.
We gradually increase the number of documents (and thus the binomial amount of queries) until the optimization is no longer able to solve the problem (i.e. achieve 100\% accuracy). We call this the \textit{critical-n} point.  

We focus on relatively small sizes for $n$, $k$,  and $d$ due to the combinatorial explosion of combinations with larger document values (i.e. 50k docs with top-$k$ of 100 gives 7.7e+311 combinations, which would be equivalent to the number of query vectors of dimension $d$ in that free embedding experiment). We use $k=2$ and increase $n$ by one for each $d$ value until it breaks. We fit a polynomial regression line to the data so we can model and extrapolate results outwards.

\begin{wrapfigure}[15]{r}{0.5\textwidth}
    \vspace{-1.5em}
    \includegraphics[width=0.46\textwidth]{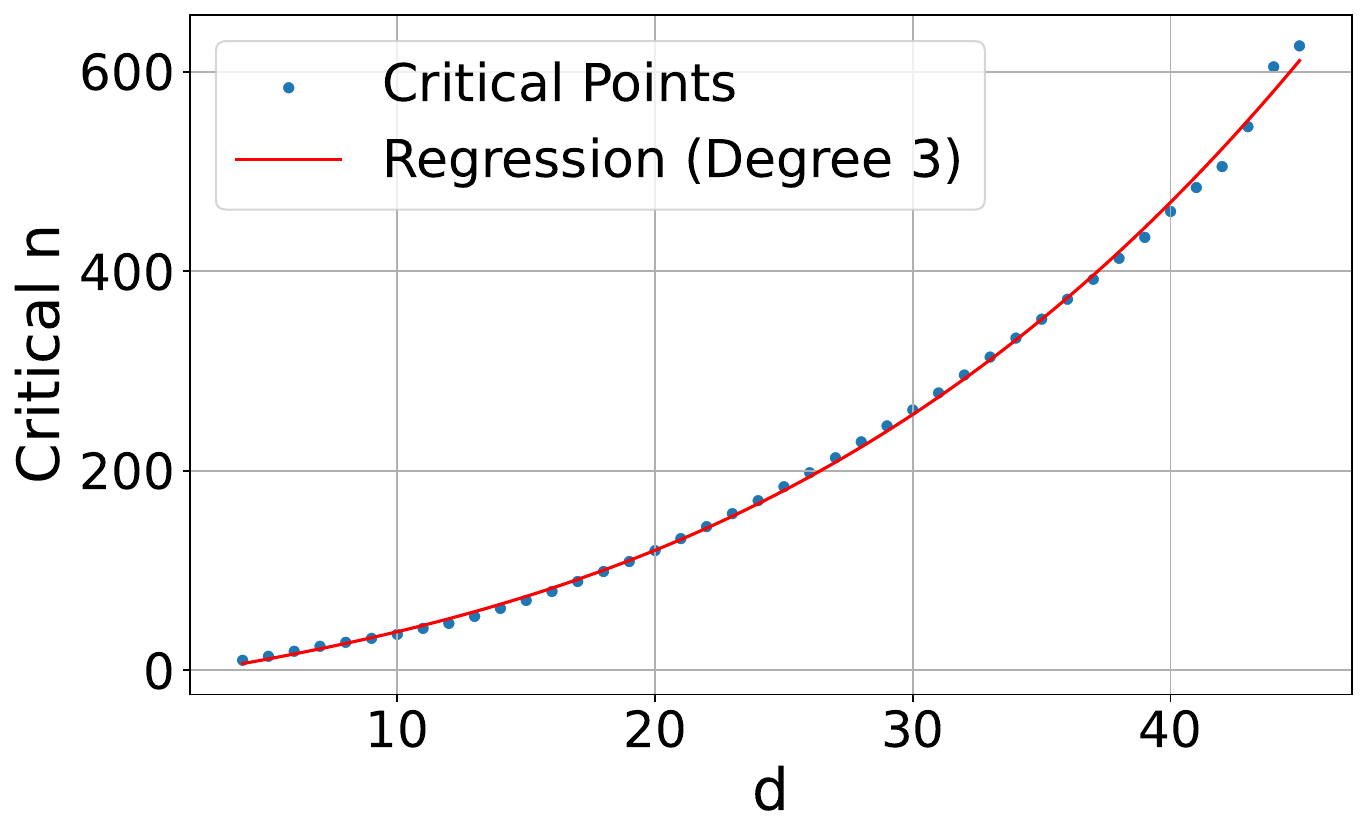}
    \captionsetup{justification=raggedright, singlelinecheck=false}
    \vspace{-0.5em}
    \caption{The critical-n value where the dimensionality is too small to successfully represent all the top-2 combinations. We plot the trend line as a polynomial function.
    }
    \label{fig:critical_n}
\end{wrapfigure}

\paragraph{Results} Figure~\ref{fig:critical_n} shows that the curve fits a 3rd degree polynomial curve, with formula $y=-10.5322 + 4.0309d + 0.0520d^2 + 0.0037d^3$ ($r^2$=0.999). Extrapolating this curve outward gives the critical-n values (for embedding size): 500k (512), 1.7m (768), 4m (1024), 107m (3072), 250m (4096). We note that this is the best case: a real embedding model cannot directly optimize the query and document vectors to match the test qrel matrix (and is constrained by factors such as "modeling natural language"). 
The results also show that the lower bounds in the previous section are a gross underestimate of real-world performance, as Table~\ref{tab:bounds} shows a lower bound of 4 for $n=100$ whereas we see the free embeddings needing $d>18$ (e.g. a 4.5 multiplier even in the no-generalization or natural language case).
Overall, these numbers already show that for web-scale search, even the largest embedding dimensions with ideal test-set optimization are not enough to model all combinations.

\section{Empirical Connection: Real-World Datasets}
\label{sec:real}
The free embedding experiments provide empirical evidence that our theoretical results hold true. However, they still are abstract - what does this mean for real embedding models? In this section we (1) draw connections from this theory to existing datasets and (2) create a trivially simple yet extremely difficult retrieval task for existing SOTA models. 

\subsection{Connection to Existing Datasets}
Existing retrieval datasets typically use a static evaluation set with limited numbers of queries, as relevance annotation is expensive to do for each query. This means practically that the space of queries used for evaluation is a very small sample of the number of potential queries. For example, the QUEST dataset \citep{malaviya2023quest} has 325k documents and queries with 20 relevant documents per query, with a total of 3357 queries. The number of unique top-20 document sets that could be returned with the QUEST corpus would be $\binom{325k}{20}$ which is equal to 7.1e+91 (larger than the estimate of atoms in the observable universe, $10^{82}$). Thus, the 3k queries in QUEST can only cover an infinitesimally small part of the qrel combination space.  

Although it is not possible to instantiate all combinations when using large-scale corpora, search evaluation datasets are a proxy for what any user would ask for and ideally would be designed to test many combinations, as users will do. In many cases, developers of new evaluations simply choose to use fewer queries due to cost or computational expense of evaluation. For example, QUEST's query "Novels from 1849 or George Sand novels" combines two categories of novels with the "OR" operator -- one could instantiate new queries to relate concepts through OR'ing other categories together. Similarly, with the rise of search agents, we see greater usage of hyper-specific queries: BrowseComp \citep{wei2025browsecomp} has 5+ conditions per query, including range operators. With these tools, it is possible to sub-select any top-$k$ relevant set with the right operators if the documents are sufficiently expressive (i.e. non-trivial). Thus, that existing datasets choose to only instantiate some of these combinations is mainly for practical reasons and not because of a lack of existence. 

In contrast to these previous works, we seek to build a dataset that evaluates all combinations of top-$k$ sets for a small number of documents. Rather than using difficult query operators like QUEST, BrowseComp, etc. (which are already difficult for reasons outside of the qrel matrix) we choose very simple queries and documents to highlight the difficulty of representing all top-$k$ sets themselves.

\subsection{The \datasetname{} Dataset}

\paragraph{Dataset Construction} In order to have a natural language version of this dataset, we need some way to map combinations of documents into something that could be retrieved with a query. One simple\footnote{This is just one way, designed to be realistic and simple. However, our framework allows for any way of instantiation -- not stuck to this arbitrary natural language design.} way to do this is to create a synthetic version with latent variables for queries and documents and then instantiate it with natural language. For this mapping, we choose to use attributes that someone could like (i.e. Jon likes Hawaiian pizza, sports cars, etc. ) as they are plentiful and don't present issues w.r.t. other items: one can like Hawaiian pizza but dislike pepperoni, all preferences are valid. We then enforce two constraints for realism: (1) users shouldn't have too many attributes, thus keeping the documents short (less than 50 per user) and (2) each query should only ask for one item to keep the task simple (i.e. "who likes X"). We gather a list of attributes a person could like through prompting Gemini 2.5 Pro. We then clean it to a final 1850 items by iteratively asking it to remove duplicates/hypernyms, while also checking the top failures with BM25 to ensure no overlap.

We choose to use 50k documents in order to have a hard but relatively small corpus and 1000 queries to maintain statistical significance while still being fast to evaluate. For each query, we choose to use two relevant documents (i.e. $k$=2), both for simplicity in instantiating and to mirror previous work (i.e. NQ, HotpotQA, etc. \citep{kwiatkowski2019natural,yang2018hotpotqa}).

 \begin{figure*}[t!]
    \includegraphics[width=0.99\linewidth,trim=0.75cm 0cm 0.75cm 0cm]{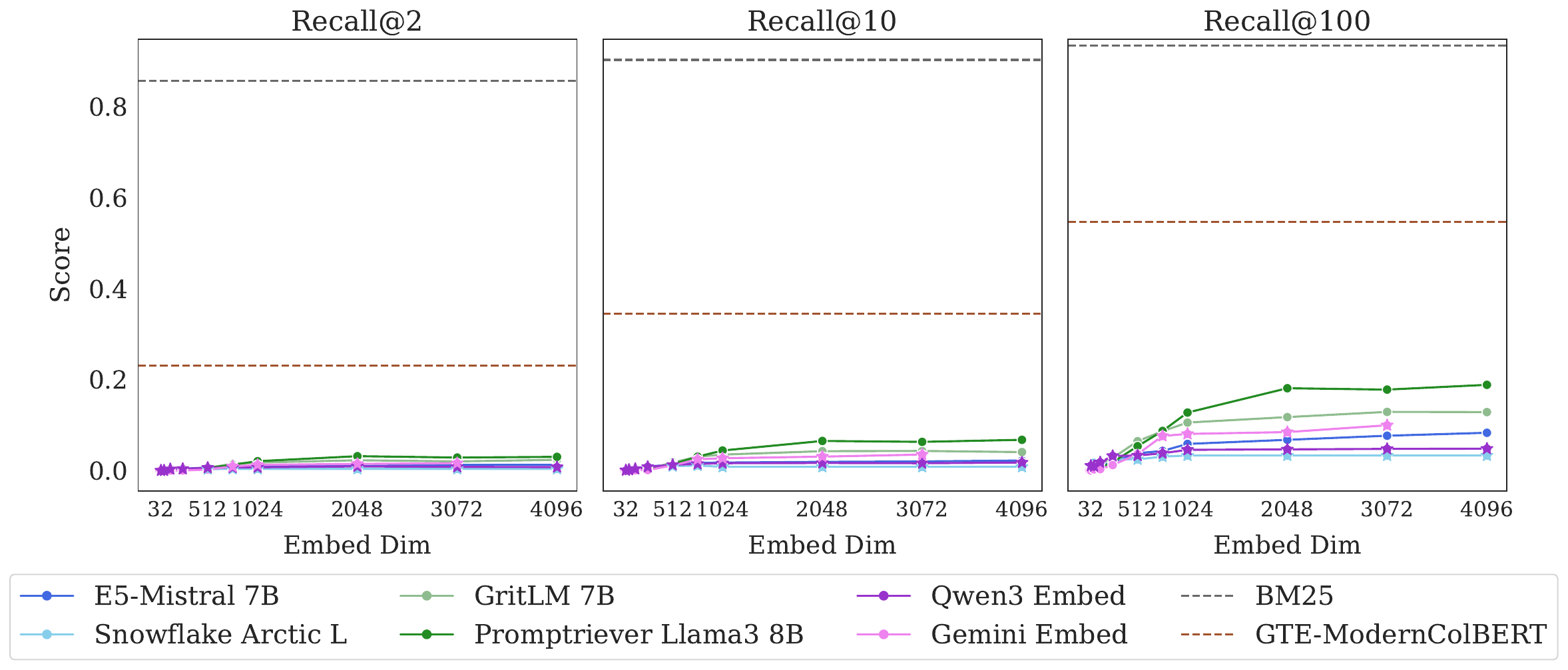}
    \captionsetup{justification=raggedright, singlelinecheck=false}
    \caption{Scores on the \datasetname{} task. Despite the simplicity of the task we see that SOTA models struggle. We also see that the dimensionality of the model is a limiting factor and that as the dimension increases, so does performance. Even multi-vector models struggle. Lexical models like BM25 do very well due to their higher dimensionality. Stars indicate models trained with MRL.}
    \vspace{-1em}
    \label{fig:real}
\end{figure*}

Our last step is to choose a qrel matrix to instantiate these attributes. Although we could not prove the hardest qrel matrix definitively with theory, we intuit that our theoretical results imply that the more interconnected the qrel matrix is (e.g. dense with all combinations) the harder it would be for models to represent.
Following this, we use the qrel matrix with the highest number of documents for which all combinations would be just above 1000 queries for a top-$k$ of 2 (46 docs, since $\binom{46}{2}$ is 1035, the smallest above 1k). 

We then assign random natural language attributes to the queries, adding these attributes to their respective relevant documents (c.f. Figure~\ref{fig:teaser}). We give each document a random first and last name from open-source lists of names.  Finally, we randomly sample new attributes for each document until all documents have the same number of attributes. As this setup has many more documents than those that are relevant to any query (46 relevant documents, 49.95k non-relevant to any query) we also create a "small" version with only the 46 documents that are relevant to one of the 1000 queries.

\textbf{Models} \hspace{0.1em} We evaluate the state-of-the-art embedding models including GritLM \citep{muennighoff2024generative}, Qwen 3 Embeddings \citep{qwen3embedding}, Promptriever \citep{weller2024promptriever}, Gemini Embeddings \citep{lee2025gemini}, Snowflake's Arctic Embed Large v2.0 \citep{yu2024arctic}, and E5-Mistral Instruct \citep{wang2022text,wang2023improving}. These models range in embedding dimension (1024 to 4096) as well as in training style (instruction-based, hard negative optimized, etc.). We also evaluate three non-single vector models to show the distinction: BM25 \citep{robertson1995okapi,lu2024bm25s}, gte-ModernColBERT \citep{GTE-ModernColBERT,PyLate}, and a token-wise TF-IDF.\footnote{This model turns each unique item into a token and then does TF-IDF. We build it to show that it gets 100\% on all tasks (as it reverse engineers our dataset construction) and thus we do not include it in future charts.} 

We show results at the full embedding dimension and also with truncated embedding dimension (typically used with matryoshka learning, aka MRL \citep{kusupati2022matryoshka}). For models not trained with MRL this will result in sub-par scores, thus, models trained with MRL are indicated with stars in the plots. However, as there are no LLMs with an embedding dimension smaller than 384, we include MRL for all models to small dimensions (32) to show the impact of embedding dimensionality.

 \begin{figure*}[t]
    \includegraphics[width=0.99\linewidth,trim=0.75cm 0cm 0.75cm 3cm]{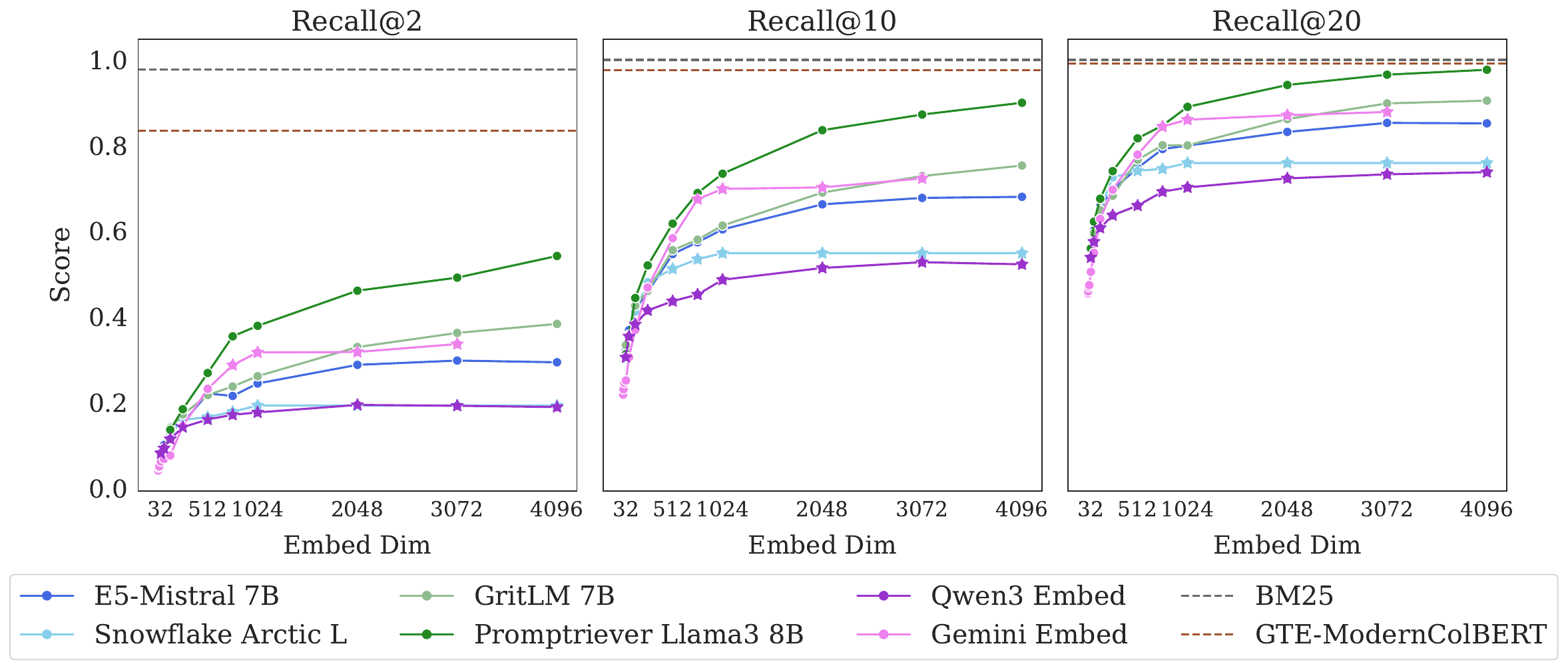}
    \captionsetup{justification=raggedright, singlelinecheck=false}
    \caption{Scores on the \datasetname\ small task (N=46) over embedding dimensions. Despite having just 46 documents, models struggle even with recall@10 and cannot solve the task even with recall@20.}
    \label{fig:real_small}
    \vspace{-1em}
\end{figure*}

\textbf{Results} \hspace{0.1em} Figure~\ref{fig:real} shows the results on the full \datasetname{} while Figure~\ref{fig:real_small} shows the results on the small (46 document) version. \textbf{The results are surprising - models severely struggle even though the task is trivially simple.} For example, in the full setting models struggle to reach even 20\% recall@100 and in the 46 document version models cannot solve the task even with recall@20.

\begin{wrapfigure}[16]{r}{0.4\textwidth}
    \vspace{-1em}
    \includegraphics[width=0.4\textwidth]{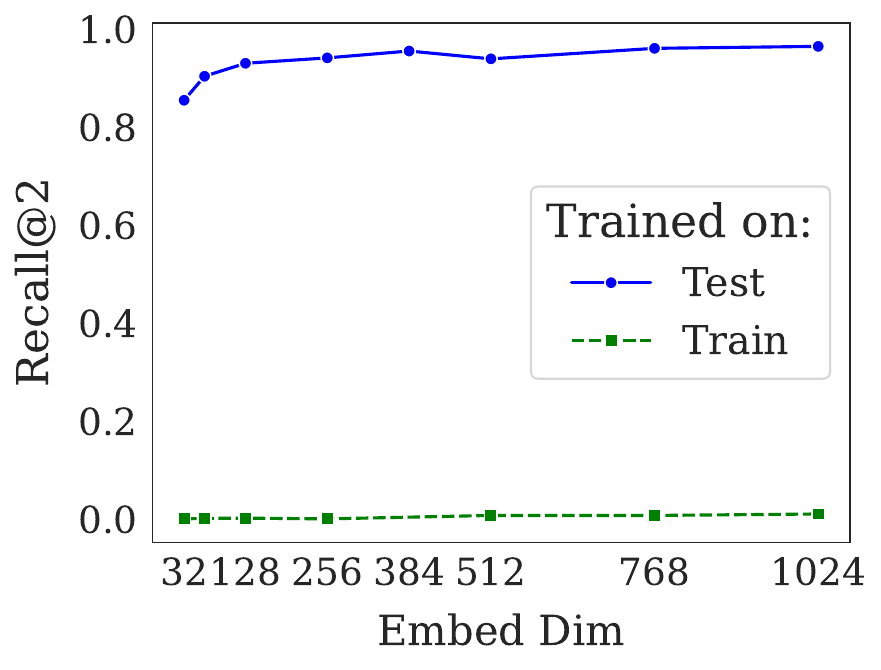}
    \captionsetup{justification=raggedright, singlelinecheck=false}
    \vspace{-1.5em}
    \caption{Training on \datasetname{} train does not significantly help, indicating the issue is not domain shift. But models can solve it if they overfit to the test set.}
    \label{fig:fine_tuning}
\end{wrapfigure}

We see that model performance depends crucially on the embedding dimensionality (better performance with bigger dimensions).  Interestingly, models trained with more diverse instruction, such as Promptriever, perform better, perhaps because their training allows them to use more of their embedding space (compared to models which are trained with MRL and on a smaller range of tasks that can perhaps be consolidated into a smaller embedding manifold). 

For alternative architectures, GTE-ModernColBERT does significantly better than single-vector models (although far from solving the task) while BM25 comes close to perfect scores. Both of these alterative architectures (sparse and multi-vector) offer various trade-offs, see \cref{sec:alternatives} for analysis.

\textbf{Is this Domain Shift?} \hspace{0.1em}
Although our queries look similar to standard web search queries, we wondered whether there could be some domain shift causing the low performance. If so, we would expect that training on a training set of similar examples would significantly improve performance. On the other hand, if the task was intrinsically hard, training on the training set would provide little help whereas training on the test set would allow the model to overfit to those tokens (similar to the free embeddings).

To test this we take an off-the-shelf embedding model and train it on either the training set (created synthetically using non-test set attributes) or the official test set of \datasetname. We use \texttt{lightonai/modernbert-embed-large} \citep{modernBERT-embed-large} and fine-tune it on these splits, using the full dataset for in batch negatives (excluding positives) using SentenceTransformers \citep{reimers-2019-sentence-bert}. We show a range of dimensions by projecting the hidden layer down to the specified size during training (rather than using MRL).

\Cref{fig:fine_tuning} shows the model trained on the training set cannot solve the problem, although it does see very minor improvement from near zero recall@10 to up to 2.8 recall@10. The lack of performance gains when training in-domain indicate that poor performance is not due to domain shift. By training the model on the test set we see it can learn the task, overfitting on the tokens in the test queries. This aligns with our free embedding results, that it is possible to overfit to the $N=46$ version with only 12 dimensions. However, it is notable that the real models with 64 dimensions still cannot completely solve the task, implying \textbf{real models perform significantly worse than the bounds shown in \S\ref{sec:free_embeds}}.

\textbf{What about Non-Lexical Matches?} \hspace{0.1em}
Our previous results show that lexical models greatly outperform their neural counterparts. However, this is not to imply that lexical models are a panacea - although they have higher dimensionality than single-vector models, they have other shortcomings. 

We illustrate this by creating a version of LIMIT-small that replaces all items in the corpus with their synonyms, reducing the amount of lexical overlap. We ask Gemini 2.5 Pro to come up with synonyms that don’t match any other existing synonyms or original items, by using either scientific names, similar meanings, or if necessary hypernyms. This creates a mapping like ``glasses” to ``spectacles”, etc.\footnote{We note that it doesn't remove all lexical overlap due to items like “Scuba Diving” -> “Underwater Diving”} We repeat the previous experiment on LIMIT-small (synonyms) and compare to LIMIT-small in Figure~\ref{fig:limit-syn}. We find that all models drop in performance as the task is now more difficult, but BM25 drops the most and now underperforms the neural models (e.g. BM25 drops more than 89\% whereas Qwen3 embedding drops 38.9\%). Thus, we can see that although lexical models have strengths like higher dimensionality, they are limited by their keyword-only matching ability. We expand more upon their strengths and weaknesses for instruction-following in Section~\ref{sec:alternatives}.

 \begin{figure*}[t]
    \includegraphics[width=0.99\linewidth,trim=0.75cm 1cm 0.75cm 4cm]{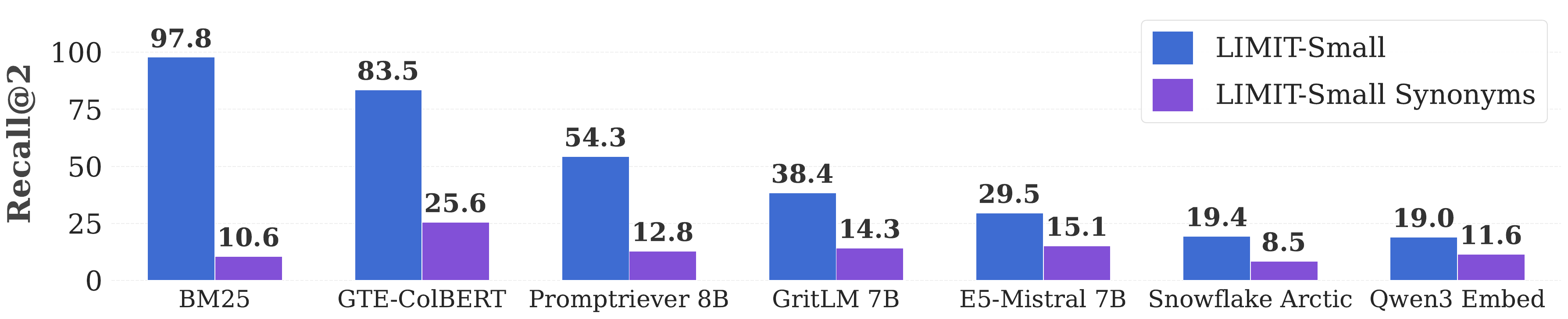}
    \captionsetup{justification=raggedright, singlelinecheck=false}
    \caption{Comparing scores on \datasetname{} small vs \datasetname{} small (synonym). Using synonyms makes the task harder so all models perform worse. However, the lexical model (BM25) drops nearly 90\%, performing worse than most single-vector models on the synonym version of \datasetname{}. Thus, lexical models have weaknesses of their own (see Section~\ref{sec:alternatives} for more discussion) and are not a panacea.}
    \label{fig:limit-syn}
    \vspace{-1em}
\end{figure*}

\textbf{Implications} \hspace{0.1em}
Single-vector models are fundamentally limited by their embedding dimension. The LIMIT dataset is a particular instantiation, with very simple queries and documents, designed to highlight this property. 
This small version of LIMIT can be embedded in just 12 dimensions (as seen in the free embeddings experiments), yet all models fail to perform well, suggesting other architectural weaknesses.
Irrespective of the architecture involved, however, our framework can scale the dataset's difficulty to consistently demonstrate this fundamental limitation.

\subsection{Alternatives to Embedding Models}
\label{sec:alternatives}
Our previous results show both theoretically and empirically that embedding models cannot represent all combinations of documents in their top-$k$ sets, making them unable to represent and solve some retrieval tasks. As current embedding models have grown larger (e.g. up to 4096), this has helped reduce negative effects for smaller dataset sizes. However, with enough combinations of top-$k$ sets the dimensionality would have to increase to an infeasible size for non-toy datasets. Thus, although they are useful for first stage results, more expressive retriever architectures will be needed.

\textbf{Cross-Encoders} \hspace{0.1em} Although not suitable for first stage retrieval at scale, they are already typically used to improve first stage results. Is \datasetname{} challenging for rerankers also? We evaluate a long context reranker, Gemini-2.5-Pro \citep{comanici2025gemini} on the small setting as a comparison. We give Gemini all 46 documents and all 1000 queries at once, asking it to output the relevant documents for each query with one generation. We find that it can successfully solve (100\%) all 1000 queries in one forward pass. This is in contrast to even the best embedding models with a recall@2 of less than 60\% (Figure~\ref{fig:real_small}). Thus we can see that \datasetname{} is easy for state-of-the-art reranker models, which do not have the same limitations based on embedding dimension. 

\textbf{Multi-vector models} \hspace{0.1em}  Multi-vector models are more expressive through the use of multiple vectors per sequence combined with the MaxSim operator \citep{khattab2020colbert}. These models show promise on the \datasetname{} dataset, with scores greatly above the single-vector models despite using a smaller backbone (ModernBERT, \citet{warner2024smarter}). However, these models are not generally used for instruction-following or reasoning-based tasks (see \citet{Reason-ModernColBERT} as one of the few that exist), leaving it an open question to how well multi-vector techniques will transfer to these tasks.

\textbf{Sparse models} \hspace{0.1em}  Sparse models (both lexical and neural) can be thought of as single vectors but with very high dimensionality. This dimensionality helps BM25 avoid the problems of the neural embedding models as seen in Figure~\ref{fig:real}. Since the $d$ of their vectors is high, they can scale to many more combinations than their dense vector counterparts. However, it is less clear how to apply sparse models to instruction-following and reasoning-based tasks where there is no lexical or even paraphrase-like overlap. We leave this direction (and hybrid sparse/dense solutions) to future work.

We note that all of these options have various trade-offs and none provide a clear path to solving this problem as-is. We leave it to future work to develop new techniques to mitigate these issues: perhaps through one of these alterative categories or through new ideas around single-vector models that can resolve the underlying issue (potentially through techniques such as hyperencoders \citep{killingback2025hypencoder} or other future work on single vector architectures yet to be developed).

\section{Conclusion}
We introduce the \datasetname\ dataset, which highlights a fundamental limitation of embedding models. We provide a theoretical connection which shows that, for a fixed embedding dimension there will be some set of documents such that certain sets are unattainable as top-$k$ sets.
We show these theoretical results hold empirically, through best case optimization of the vectors themselves, and make a practical connection to existing state-of-the-art models by creating a realistic and simple instantiation of the theory, called \datasetname{}, that these models cannot solve. 
Our results imply that the community should reconsider how instruction-based retrieval will impact future retrievers.

\section*{Limitations}
Although our experiments provide theoretical insight for the most common type of embedding model (single vector) they do not hold necessarily for other architectures, such as multi-vector models. Although we showed initial empirical results with non-single vector models, we leave it to future work to extend our theoretical connections to these settings. We also did not show theoretical results for the setting where the user allows some mistakes, e.g. capturing only the majority of the combinations. We leave putting a bound on this scenario to future work and would invite the reader to examine works like \citet{ben2002limitations}. 

We have shown the theoretical connection that proves that some combinations cannot be represented by embedding models, however, we cannot prove apriori which \textit{types} of combinations they will fail on. Thus, it is possible that there are some instruction-following or reasoning tasks they can solve perfectly, however, \textit{we do know} that there exists some tasks that they will never be able to solve.

\section*{Acknowledgments} We thank Tanmaya Dabral, Zhongli Ding, Anthony Chen, Ming-Wei Chang, Kenton Lee, and Kristina Toutanova for their helpful feedback. We thank Kiril Bangachev, Guy Bresler, Iliyas Noman, Yury Polyanskiy, Antonio Vergari, Adam Lopez, and Andreas Grivas for pointers to work on sign-rank.

\bibliography{template_refs}

\begin{thebibliography}{73}
\providecommand{\natexlab}[1]{#1}
\providecommand{\url}[1]{\texttt{#1}}
\expandafter\ifx\csname urlstyle\endcsname\relax
  \providecommand{\doi}[1]{doi: #1}\else
  \providecommand{\doi}{doi: \begingroup \urlstyle{rm}\Url}\fi

\bibitem[Alon et~al.(1985)Alon, Frankl, and Rodl]{alongeometrical}
Noga Alon, Peter Frankl, and Vojtech Rodl.
\newblock Geometrical realization of set systems and probabilistic communication complexity.
\newblock In \emph{26th Annual Symposium on Foundations of Computer Science (sfcs 1985)}, pp.\  277--280. IEEE, 1985.

\bibitem[Badreddine et~al.(2025)Badreddine, van Krieken, and Serafini]{badreddine2025breaking}
Samy Badreddine, Emile van Krieken, and Luciano Serafini.
\newblock Breaking rank bottlenecks in knowledge graph embeddings.
\newblock \emph{arXiv preprint arXiv:2506.22271}, 2025.

\bibitem[Bangachev et~al.(2025)Bangachev, Bresler, Noman, and Polyanskiy]{bangachev2025global}
Kiril Bangachev, Guy Bresler, Iliyas Noman, and Yury Polyanskiy.
\newblock Global minimizers of sigmoid contrastive loss.
\newblock \emph{arXiv preprint arXiv:2509.18552}, 2025.

\bibitem[Bartlett(2002)]{bartlett1998sample}
Peter~L Bartlett.
\newblock The sample complexity of pattern classification with neural networks: the size of the weights is more important than the size of the network.
\newblock \emph{IEEE transactions on Information Theory}, 44\penalty0 (2):\penalty0 525--536, 2002.

\bibitem[BehnamGhader et~al.(2024)BehnamGhader, Adlakha, Mosbach, Bahdanau, Chapados, and Reddy]{behnamghader2024llm2vec}
Parishad BehnamGhader, Vaibhav Adlakha, Marius Mosbach, Dzmitry Bahdanau, Nicolas Chapados, and Siva Reddy.
\newblock Llm2vec: Large language models are secretly powerful text encoders.
\newblock \emph{arXiv preprint arXiv:2404.05961}, 2024.

\bibitem[Ben-David et~al.(2002)Ben-David, Eiron, and Simon]{ben2002limitations}
Shai Ben-David, Nadav Eiron, and Hans~Ulrich Simon.
\newblock Limitations of learning via embeddings in euclidean half spaces.
\newblock \emph{Journal of Machine Learning Research}, 3\penalty0 (Nov):\penalty0 441--461, 2002.

\bibitem[Bohler et~al.(2015)Bohler, Cheilaris, Klein, Liu, Papadopoulou, and Zavershynskyi]{BOHLER2015539}
Cecilia Bohler, Panagiotis Cheilaris, Rolf Klein, Chih-Hung Liu, Evanthia Papadopoulou, and Maksym Zavershynskyi.
\newblock On the complexity of higher order abstract voronoi diagrams.
\newblock \emph{Computational Geometry}, 48\penalty0 (8):\penalty0 539--551, 2015.
\newblock ISSN 0925-7721.
\newblock \doi{https://doi.org/10.1016/j.comgeo.2015.04.008}.
\newblock URL \url{https://www.sciencedirect.com/science/article/pii/S0925772115000346}.

\bibitem[Chaffin(2025{\natexlab{a}})]{GTE-ModernColBERT}
Antoine Chaffin.
\newblock Gte-moderncolbert, 2025{\natexlab{a}}.
\newblock URL \url{https://huggingface.co/lightonai/GTE-ModernColBERT-v1}.

\bibitem[Chaffin(2025{\natexlab{b}})]{Reason-ModernColBERT}
Antoine Chaffin.
\newblock Reason-moderncolbert, 2025{\natexlab{b}}.
\newblock URL \url{https://huggingface.co/lightonai/Reason-ModernColBERT}.

\bibitem[Chaffin(2025{\natexlab{c}})]{modernBERT-embed-large}
Antoine Chaffin.
\newblock Modernbert-embed-large, 2025{\natexlab{c}}.
\newblock URL \url{https://huggingface.co/lightonai/modernbert-embed-large}.

\bibitem[Chaffin \& Sourty(2024)Chaffin and Sourty]{PyLate}
Antoine Chaffin and Raphaël Sourty.
\newblock Pylate: Flexible training and retrieval for late interaction models, 2024.
\newblock URL \url{https://github.com/lightonai/pylate}.

\bibitem[Chen et~al.(2023)Chen, Huang, Chen, Liu, Chen, and Jia]{chen2023efficient}
Bi~Yu Chen, Huihuang Huang, Hui-Ping Chen, Wenxuan Liu, Xuan-Yan Chen, and Tao Jia.
\newblock Efficient algorithm for constructing order k voronoi diagrams in road networks.
\newblock \emph{ISPRS International Journal of Geo-Information}, 12\penalty0 (4):\penalty0 172, 2023.

\bibitem[Clarkson(1988)]{clarkson1988applications}
Kenneth~L Clarkson.
\newblock Applications of random sampling in computational geometry, ii.
\newblock In \emph{Proceedings of the fourth annual symposium on Computational geometry}, pp.\  1--11, 1988.

\bibitem[Comanici et~al.(2025)Comanici, Bieber, Schaekermann, Pasupat, Sachdeva, Dhillon, Blistein, Ram, Zhang, Rosen, et~al.]{comanici2025gemini}
Gheorghe Comanici, Eric Bieber, Mike Schaekermann, Ice Pasupat, Noveen Sachdeva, Inderjit Dhillon, Marcel Blistein, Ori Ram, Dan Zhang, Evan Rosen, et~al.
\newblock Gemini 2.5: Pushing the frontier with advanced reasoning, multimodality, long context, and next generation agentic capabilities.
\newblock \emph{arXiv preprint arXiv:2507.06261}, 2025.

\bibitem[Conway et~al.(1999)Conway, Goodman-Strauss, and Sloane]{conway1999recent}
John~H Conway, Chaim Goodman-Strauss, and N~Sloane.
\newblock Recent progress in sphere packing.
\newblock \emph{Current Developments in Mathematics}, 1999\penalty0 (1):\penalty0 37--76, 1999.

\bibitem[Craswell et~al.(2020)Craswell, Mitra, Yilmaz, Campos, and Voorhees]{craswell2020overview}
Nick Craswell, Bhaskar Mitra, Emine Yilmaz, Daniel Campos, and Ellen~M Voorhees.
\newblock Overview of the trec 2019 deep learning track.
\newblock \emph{arXiv preprint arXiv:2003.07820}, 2020.

\bibitem[Enevoldsen et~al.(2025)Enevoldsen, Chung, Kerboua, Kardos, Mathur, Stap, Gala, Siblini, Krzemi{\'n}ski, Winata, et~al.]{enevoldsen2025mmteb}
Kenneth Enevoldsen, Isaac Chung, Imene Kerboua, M{\'a}rton Kardos, Ashwin Mathur, David Stap, Jay Gala, Wissam Siblini, Dominik Krzemi{\'n}ski, Genta~Indra Winata, et~al.
\newblock Mmteb: Massive multilingual text embedding benchmark.
\newblock \emph{arXiv preprint arXiv:2502.13595}, 2025.

\bibitem[Faysse et~al.(2024)Faysse, Sibille, Wu, Omrani, Viaud, Hudelot, and Colombo]{faysse2024colpali}
Manuel Faysse, Hugues Sibille, Tony Wu, Bilel Omrani, Gautier Viaud, C{\'e}line Hudelot, and Pierre Colombo.
\newblock Colpali: Efficient document retrieval with vision language models.
\newblock \emph{arXiv preprint arXiv:2407.01449}, 2024.

\bibitem[Grivas et~al.(2024)Grivas, Vergari, and Lopez]{grivas2024taming}
Andreas Grivas, Antonio Vergari, and Adam Lopez.
\newblock Taming the sigmoid bottleneck: Provably argmaxable sparse multi-label classification.
\newblock In \emph{Proceedings of the AAAI Conference on Artificial Intelligence}, volume~38, pp.\  12208--12216, 2024.

\bibitem[Hoffmann et~al.(2022)Hoffmann, Borgeaud, Mensch, Buchatskaya, Cai, Rutherford, Casas, Hendricks, Welbl, Clark, et~al.]{hoffmann2022training}
Jordan Hoffmann, Sebastian Borgeaud, Arthur Mensch, Elena Buchatskaya, Trevor Cai, Eliza Rutherford, Diego de~Las Casas, Lisa~Anne Hendricks, Johannes Welbl, Aidan Clark, et~al.
\newblock Training compute-optimal large language models.
\newblock \emph{arXiv preprint arXiv:2203.15556}, 2022.

\bibitem[Izacard et~al.(2021)Izacard, Caron, Hosseini, Riedel, Bojanowski, Joulin, and Grave]{izacard2021unsupervised}
Gautier Izacard, Mathilde Caron, Lucas Hosseini, Sebastian Riedel, Piotr Bojanowski, Armand Joulin, and Edouard Grave.
\newblock Unsupervised dense information retrieval with contrastive learning.
\newblock \emph{arXiv preprint arXiv:2112.09118}, 2021.

\bibitem[Johnson et~al.(1984)Johnson, Lindenstrauss, et~al.]{johnson1984extensions}
William~B Johnson, Joram Lindenstrauss, et~al.
\newblock Extensions of lipschitz mappings into a hilbert space.
\newblock \emph{Contemporary mathematics}, 26\penalty0 (189-206):\penalty0 1, 1984.

\bibitem[Kearns \& Schapire(1994)Kearns and Schapire]{kearns1994efficient}
Michael~J Kearns and Robert~E Schapire.
\newblock Efficient distribution-free learning of probabilistic concepts.
\newblock \emph{Journal of Computer and System Sciences}, 48\penalty0 (3):\penalty0 464--497, 1994.

\bibitem[Khattab \& Zaharia(2020)Khattab and Zaharia]{khattab2020colbert}
Omar Khattab and Matei Zaharia.
\newblock Colbert: Efficient and effective passage search via contextualized late interaction over bert.
\newblock In \emph{Proceedings of the 43rd International ACM SIGIR conference on research and development in Information Retrieval}, pp.\  39--48, 2020.

\bibitem[Killingback et~al.(2025)Killingback, Zeng, and Zamani]{killingback2025hypencoder}
Julian Killingback, Hansi Zeng, and Hamed Zamani.
\newblock Hypencoder: Hypernetworks for information retrieval.
\newblock In \emph{Proceedings of the 48th International ACM SIGIR Conference on Research and Development in Information Retrieval}, pp.\  2372--2383, 2025.

\bibitem[Kingma \& Ba(2014)Kingma and Ba]{kingma2014adam}
Diederik~P Kingma and Jimmy Ba.
\newblock Adam: A method for stochastic optimization.
\newblock \emph{arXiv preprint arXiv:1412.6980}, 2014.

\bibitem[Kusupati et~al.(2022)Kusupati, Bhatt, Rege, Wallingford, Sinha, Ramanujan, Howard-Snyder, Chen, Kakade, Jain, et~al.]{kusupati2022matryoshka}
Aditya Kusupati, Gantavya Bhatt, Aniket Rege, Matthew Wallingford, Aditya Sinha, Vivek Ramanujan, William Howard-Snyder, Kaifeng Chen, Sham Kakade, Prateek Jain, et~al.
\newblock Matryoshka representation learning.
\newblock \emph{Advances in Neural Information Processing Systems}, 35:\penalty0 30233--30249, 2022.

\bibitem[Kwiatkowski et~al.(2019)Kwiatkowski, Palomaki, Redfield, Collins, Parikh, Alberti, Epstein, Polosukhin, Devlin, Lee, et~al.]{kwiatkowski2019natural}
Tom Kwiatkowski, Jennimaria Palomaki, Olivia Redfield, Michael Collins, Ankur Parikh, Chris Alberti, Danielle Epstein, Illia Polosukhin, Jacob Devlin, Kenton Lee, et~al.
\newblock Natural questions: a benchmark for question answering research.
\newblock \emph{Transactions of the Association for Computational Linguistics}, 7:\penalty0 453--466, 2019.

\bibitem[Lee(1982)]{lee1982k}
Der-Tsai Lee.
\newblock On k-nearest neighbor voronoi diagrams in the plane.
\newblock \emph{IEEE transactions on computers}, 100\penalty0 (6):\penalty0 478--487, 1982.

\bibitem[Lee et~al.(2024)Lee, Dai, Ren, Chen, Cer, Cole, Hui, Boratko, Kapadia, Ding, et~al.]{lee2024gecko}
Jinhyuk Lee, Zhuyun Dai, Xiaoqi Ren, Blair Chen, Daniel Cer, Jeremy~R Cole, Kai Hui, Michael Boratko, Rajvi Kapadia, Wen Ding, et~al.
\newblock Gecko: Versatile text embeddings distilled from large language models.
\newblock \emph{arXiv preprint arXiv:2403.20327}, 2024.

\bibitem[Lee et~al.(2025)Lee, Chen, Dua, Cer, Shanbhogue, Naim, {\'A}brego, Li, Chen, Vera, et~al.]{lee2025gemini}
Jinhyuk Lee, Feiyang Chen, Sahil Dua, Daniel Cer, Madhuri Shanbhogue, Iftekhar Naim, Gustavo~Hern{\'a}ndez {\'A}brego, Zhe Li, Kaifeng Chen, Henrique~Schechter Vera, et~al.
\newblock Gemini embedding: Generalizable embeddings from gemini.
\newblock \emph{arXiv preprint arXiv:2503.07891}, 2025.

\bibitem[Lee et~al.(2019)Lee, Chang, and Toutanova]{lee-etal-2019-latent}
Kenton Lee, Ming-Wei Chang, and Kristina Toutanova.
\newblock Latent retrieval for weakly supervised open domain question answering.
\newblock In Anna Korhonen, David Traum, and Llu{\'i}s M{\`a}rquez (eds.), \emph{Proceedings of the 57th Annual Meeting of the Association for Computational Linguistics}, pp.\  6086--6096, Florence, Italy, July 2019. Association for Computational Linguistics.
\newblock \doi{10.18653/v1/P19-1612}.
\newblock URL \url{https://aclanthology.org/P19-1612/}.

\bibitem[Li et~al.(2024)Li, Gan, Yang, Yang, Li, Wang, Gao, et~al.]{li2024multimodal}
Chunyuan Li, Zhe Gan, Zhengyuan Yang, Jianwei Yang, Linjie Li, Lijuan Wang, Jianfeng Gao, et~al.
\newblock Multimodal foundation models: From specialists to general-purpose assistants.
\newblock \emph{Foundations and Trends{\textregistered} in Computer Graphics and Vision}, 16\penalty0 (1-2):\penalty0 1--214, 2024.

\bibitem[L{\`u}(2024)]{lu2024bm25s}
Xing~Han L{\`u}.
\newblock Bm25s: Orders of magnitude faster lexical search via eager sparse scoring.
\newblock \emph{arXiv preprint arXiv:2407.03618}, 2024.

\bibitem[Ma et~al.(2024)Ma, Lin, Li, Chen, and Lin]{ma2024unifying}
Xueguang Ma, Sheng-Chieh Lin, Minghan Li, Wenhu Chen, and Jimmy Lin.
\newblock Unifying multimodal retrieval via document screenshot embedding.
\newblock \emph{arXiv preprint arXiv:2406.11251}, 2024.

\bibitem[Malaviya et~al.(2023)Malaviya, Shaw, Chang, Lee, and Toutanova]{malaviya2023quest}
Chaitanya Malaviya, Peter Shaw, Ming-Wei Chang, Kenton Lee, and Kristina Toutanova.
\newblock Quest: A retrieval dataset of entity-seeking queries with implicit set operations.
\newblock \emph{arXiv preprint arXiv:2305.11694}, 2023.

\bibitem[Muennighoff et~al.(2022)Muennighoff, Tazi, Magne, and Reimers]{muennighoff2022mteb}
Niklas Muennighoff, Nouamane Tazi, Lo{\"\i}c Magne, and Nils Reimers.
\newblock Mteb: Massive text embedding benchmark.
\newblock \emph{arXiv preprint arXiv:2210.07316}, 2022.

\bibitem[Muennighoff et~al.(2024)Muennighoff, Hongjin, Wang, Yang, Wei, Yu, Singh, and Kiela]{muennighoff2024generative}
Niklas Muennighoff, SU~Hongjin, Liang Wang, Nan Yang, Furu Wei, Tao Yu, Amanpreet Singh, and Douwe Kiela.
\newblock Generative representational instruction tuning.
\newblock In \emph{ICLR 2024 Workshop: How Far Are We From AGI}, 2024.

\bibitem[Oh et~al.(2024)Oh, Lee, Ye, Shin, Jang, Jun, and Seo]{oh2024instructir}
Hanseok Oh, Hyunji Lee, Seonghyeon Ye, Haebin Shin, Hansol Jang, Changwook Jun, and Minjoon Seo.
\newblock Instructir: A benchmark for instruction following of information retrieval models.
\newblock \emph{arXiv preprint arXiv:2402.14334}, 2024.

\bibitem[Oord et~al.(2018)Oord, Li, and Vinyals]{oord2018representation}
Aaron van~den Oord, Yazhe Li, and Oriol Vinyals.
\newblock Representation learning with contrastive predictive coding.
\newblock \emph{arXiv preprint arXiv:1807.03748}, 2018.

\bibitem[Ormazabal et~al.(2019)Ormazabal, Artetxe, Labaka, Soroa, and Agirre]{ormazabal2019analyzing}
Aitor Ormazabal, Mikel Artetxe, Gorka Labaka, Aitor Soroa, and Eneko Agirre.
\newblock Analyzing the limitations of cross-lingual word embedding mappings.
\newblock \emph{arXiv preprint arXiv:1906.05407}, 2019.

\bibitem[Ouyang et~al.(2022)Ouyang, Wu, Jiang, Almeida, Wainwright, Mishkin, Zhang, Agarwal, Slama, Ray, et~al.]{ouyang2022training}
Long Ouyang, Jeffrey Wu, Xu~Jiang, Diogo Almeida, Carroll Wainwright, Pamela Mishkin, Chong Zhang, Sandhini Agarwal, Katarina Slama, Alex Ray, et~al.
\newblock Training language models to follow instructions with human feedback.
\newblock \emph{Advances in neural information processing systems}, 35:\penalty0 27730--27744, 2022.

\bibitem[Papadimitriou \& Sipser(1982)Papadimitriou and Sipser]{papadimitriou1982communication}
Christos~H Papadimitriou and Michael Sipser.
\newblock Communication complexity.
\newblock In \emph{Proceedings of the fourteenth annual ACM symposium on Theory of computing}, pp.\  196--200, 1982.

\bibitem[Paul et~al.(2021)Paul, Chang, and McCallum]{paul2021multi}
Rohan Paul, Haw-Shiuan Chang, and Andrew McCallum.
\newblock Multi-facet universal schema.
\newblock In \emph{Proceedings of the 16th Conference of the European Chapter of the Association for Computational Linguistics: Main Volume}, pp.\  909--919, 2021.

\bibitem[Reimers \& Gurevych(2019)Reimers and Gurevych]{reimers-2019-sentence-bert}
Nils Reimers and Iryna Gurevych.
\newblock Sentence-bert: Sentence embeddings using siamese bert-networks.
\newblock In \emph{Proceedings of the 2019 Conference on Empirical Methods in Natural Language Processing}. Association for Computational Linguistics, 11 2019.
\newblock URL \url{https://arxiv.org/abs/1908.10084}.

\bibitem[Reimers \& Gurevych(2020)Reimers and Gurevych]{reimers2020curse}
Nils Reimers and Iryna Gurevych.
\newblock The curse of dense low-dimensional information retrieval for large index sizes.
\newblock \emph{arXiv preprint arXiv:2012.14210}, 2020.

\bibitem[Robertson et~al.(1995)Robertson, Walker, Jones, Hancock-Beaulieu, Gatford, et~al.]{robertson1995okapi}
Stephen~E Robertson, Steve Walker, Susan Jones, Micheline~M Hancock-Beaulieu, Mike Gatford, et~al.
\newblock Okapi at trec-3.
\newblock \emph{Nist Special Publication Sp}, 109:\penalty0 109, 1995.

\bibitem[Song et~al.(2025)Song, Gan, Shang, and Zhao]{song2025ifir}
Tingyu Song, Guo Gan, Mingsheng Shang, and Yilun Zhao.
\newblock Ifir: A comprehensive benchmark for evaluating instruction-following in expert-domain information retrieval.
\newblock \emph{arXiv preprint arXiv:2503.04644}, 2025.

\bibitem[Su et~al.(2022)Su, Shi, Kasai, Wang, Hu, Ostendorf, Yih, Smith, Zettlemoyer, and Yu]{su2022one}
Hongjin Su, Weijia Shi, Jungo Kasai, Yizhong Wang, Yushi Hu, Mari Ostendorf, Wen-tau Yih, Noah~A Smith, Luke Zettlemoyer, and Tao Yu.
\newblock One embedder, any task: Instruction-finetuned text embeddings.
\newblock \emph{arXiv preprint arXiv:2212.09741}, 2022.

\bibitem[Su et~al.(2024)Su, Yen, Xia, Shi, Muennighoff, Wang, Liu, Shi, Siegel, Tang, et~al.]{su2024bright}
Hongjin Su, Howard Yen, Mengzhou Xia, Weijia Shi, Niklas Muennighoff, Han-yu Wang, Haisu Liu, Quan Shi, Zachary~S Siegel, Michael Tang, et~al.
\newblock Bright: A realistic and challenging benchmark for reasoning-intensive retrieval.
\newblock \emph{arXiv preprint arXiv:2407.12883}, 2024.

\bibitem[Team(2024)]{team2024chameleon}
Chameleon Team.
\newblock Chameleon: Mixed-modal early-fusion foundation models.
\newblock \emph{arXiv preprint arXiv:2405.09818}, 2024.

\bibitem[Thakur et~al.(2021)Thakur, Reimers, R{\"u}ckl{\'e}, Srivastava, and Gurevych]{thakur2021beir}
Nandan Thakur, Nils Reimers, Andreas R{\"u}ckl{\'e}, Abhishek Srivastava, and Iryna Gurevych.
\newblock Beir: A heterogenous benchmark for zero-shot evaluation of information retrieval models.
\newblock \emph{arXiv preprint arXiv:2104.08663}, 2021.

\bibitem[Thakur et~al.(2025)Thakur, Lin, Havens, Carbin, Khattab, and Drozdov]{thakur2025freshstack}
Nandan Thakur, Jimmy Lin, Sam Havens, Michael Carbin, Omar Khattab, and Andrew Drozdov.
\newblock Freshstack: Building realistic benchmarks for evaluating retrieval on technical documents.
\newblock \emph{arXiv preprint arXiv:2504.13128}, 2025.

\bibitem[Vapnik(1998)]{vapnik1998statistical}
Vladimir Vapnik.
\newblock \emph{Statistical Learning Theory now plays a more active role: after the general analysis of learning processes, the research in the area of synthesis of optimal algorithms was started. These studies, however, do not belong to history yet. They are a subject of today's research activities.}
\newblock PhD thesis, These studies, however, do not belong to history yet. They are a subject of~…, 1998.

\bibitem[Vershynin(2018)]{vershynin2018high}
Roman Vershynin.
\newblock \emph{High-dimensional probability: An introduction with applications in data science}, volume~47.
\newblock Cambridge university press, 2018.

\bibitem[Voronoi(1908)]{voronoi1908nouvelles}
Georges Voronoi.
\newblock Nouvelles applications des param{\`e}tres continus {\`a} la th{\'e}orie des formes quadratiques. deuxi{\`e}me m{\'e}moire. recherches sur les parall{\'e}llo{\`e}dres primitifs.
\newblock \emph{Journal f{\"u}r die reine und angewandte Mathematik (Crelles Journal)}, 1908\penalty0 (134):\penalty0 198--287, 1908.

\bibitem[Wadden et~al.(2020)Wadden, Lin, Lo, Wang, van Zuylen, Cohan, and Hajishirzi]{wadden2020fact}
David Wadden, Shanchuan Lin, Kyle Lo, Lucy~Lu Wang, Madeleine van Zuylen, Arman Cohan, and Hannaneh Hajishirzi.
\newblock Fact or fiction: Verifying scientific claims.
\newblock \emph{arXiv preprint arXiv:2004.14974}, 2020.

\bibitem[Wang et~al.(2022)Wang, Yang, Huang, Jiao, Yang, Jiang, Majumder, and Wei]{wang2022text}
Liang Wang, Nan Yang, Xiaolong Huang, Binxing Jiao, Linjun Yang, Daxin Jiang, Rangan Majumder, and Furu Wei.
\newblock Text embeddings by weakly-supervised contrastive pre-training.
\newblock \emph{arXiv preprint arXiv:2212.03533}, 2022.

\bibitem[Wang et~al.(2023)Wang, Yang, Huang, Yang, Majumder, and Wei]{wang2023improving}
Liang Wang, Nan Yang, Xiaolong Huang, Linjun Yang, Rangan Majumder, and Furu Wei.
\newblock Improving text embeddings with large language models.
\newblock \emph{arXiv preprint arXiv:2401.00368}, 2023.

\bibitem[Warner et~al.(2024)Warner, Chaffin, Clavi{\'e}, Weller, Hallstr{\"o}m, Taghadouini, Gallagher, Biswas, Ladhak, Aarsen, et~al.]{warner2024smarter}
Benjamin Warner, Antoine Chaffin, Benjamin Clavi{\'e}, Orion Weller, Oskar Hallstr{\"o}m, Said Taghadouini, Alexis Gallagher, Raja Biswas, Faisal Ladhak, Tom Aarsen, et~al.
\newblock Smarter, better, faster, longer: A modern bidirectional encoder for fast, memory efficient, and long context finetuning and inference.
\newblock \emph{arXiv preprint arXiv:2412.13663}, 2024.

\bibitem[Wei et~al.(2025)Wei, Sun, Papay, McKinney, Han, Fulford, Chung, Passos, Fedus, and Glaese]{wei2025browsecomp}
Jason Wei, Zhiqing Sun, Spencer Papay, Scott McKinney, Jeffrey Han, Isa Fulford, Hyung~Won Chung, Alex~Tachard Passos, William Fedus, and Amelia Glaese.
\newblock Browsecomp: A simple yet challenging benchmark for browsing agents.
\newblock \emph{arXiv preprint arXiv:2504.12516}, 2025.

\bibitem[Weller et~al.(2024{\natexlab{a}})Weller, Chang, MacAvaney, Lo, Cohan, Van~Durme, Lawrie, and Soldaini]{weller2024followir}
Orion Weller, Benjamin Chang, Sean MacAvaney, Kyle Lo, Arman Cohan, Benjamin Van~Durme, Dawn Lawrie, and Luca Soldaini.
\newblock Followir: Evaluating and teaching information retrieval models to follow instructions.
\newblock \emph{arXiv preprint arXiv:2403.15246}, 2024{\natexlab{a}}.

\bibitem[Weller et~al.(2024{\natexlab{b}})Weller, Van~Durme, Lawrie, Paranjape, Zhang, and Hessel]{weller2024promptriever}
Orion Weller, Benjamin Van~Durme, Dawn Lawrie, Ashwin Paranjape, Yuhao Zhang, and Jack Hessel.
\newblock Promptriever: Instruction-trained retrievers can be prompted like language models.
\newblock \emph{arXiv preprint arXiv:2409.11136}, 2024{\natexlab{b}}.

\bibitem[Weller et~al.(2025{\natexlab{a}})Weller, Chang, Yang, Yarmohammadi, Barham, MacAvaney, Cohan, Soldaini, Van~Durme, and Lawrie]{weller2025mfollowir}
Orion Weller, Benjamin Chang, Eugene Yang, Mahsa Yarmohammadi, Sam Barham, Sean MacAvaney, Arman Cohan, Luca Soldaini, Benjamin Van~Durme, and Dawn Lawrie.
\newblock mfollowir: a multilingual benchmark for instruction following in retrieval.
\newblock \emph{arXiv preprint arXiv:2501.19264}, 2025{\natexlab{a}}.

\bibitem[Weller et~al.(2025{\natexlab{b}})Weller, Ricci, Marone, Chaffin, Lawrie, and Van~Durme]{weller2025seq}
Orion Weller, Kathryn Ricci, Marc Marone, Antoine Chaffin, Dawn Lawrie, and Benjamin Van~Durme.
\newblock Seq vs seq: An open suite of paired encoders and decoders.
\newblock \emph{arXiv preprint arXiv:2507.11412}, 2025{\natexlab{b}}.

\bibitem[Weller et~al.(2025{\natexlab{c}})Weller, Ricci, Yang, Yates, Lawrie, and Van~Durme]{weller2025rank1}
Orion Weller, Kathryn Ricci, Eugene Yang, Andrew Yates, Dawn Lawrie, and Benjamin Van~Durme.
\newblock Rank1: Test-time compute for reranking in information retrieval.
\newblock \emph{arXiv preprint arXiv:2502.18418}, 2025{\natexlab{c}}.

\bibitem[Xiao et~al.(2024)Xiao, Hudson, and Moubayed]{xiao2024rar}
Chenghao Xiao, G~Thomas Hudson, and Noura~Al Moubayed.
\newblock Rar-b: Reasoning as retrieval benchmark.
\newblock \emph{arXiv preprint arXiv:2404.06347}, 2024.

\bibitem[Yang et~al.(2018)Yang, Qi, Zhang, Bengio, Cohen, Salakhutdinov, and Manning]{yang2018hotpotqa}
Zhilin Yang, Peng Qi, Saizheng Zhang, Yoshua Bengio, William~W Cohen, Ruslan Salakhutdinov, and Christopher~D Manning.
\newblock Hotpotqa: A dataset for diverse, explainable multi-hop question answering.
\newblock \emph{arXiv preprint arXiv:1809.09600}, 2018.

\bibitem[Yin \& Shen(2018)Yin and Shen]{yin2018dimensionality}
Zi~Yin and Yuanyuan Shen.
\newblock On the dimensionality of word embedding.
\newblock \emph{Advances in neural information processing systems}, 31, 2018.

\bibitem[Yu et~al.(2024)Yu, Merrick, Nuti, and Campos]{yu2024arctic}
Puxuan Yu, Luke Merrick, Gaurav Nuti, and Daniel Campos.
\newblock Arctic-embed 2.0: Multilingual retrieval without compromise.
\newblock \emph{arXiv preprint arXiv:2412.04506}, 2024.

\bibitem[Zhang et~al.(2025)Zhang, Li, Long, Zhang, Lin, Yang, Xie, Yang, Liu, Lin, Huang, and Zhou]{qwen3embedding}
Yanzhao Zhang, Mingxin Li, Dingkun Long, Xin Zhang, Huan Lin, Baosong Yang, Pengjun Xie, An~Yang, Dayiheng Liu, Junyang Lin, Fei Huang, and Jingren Zhou.
\newblock Qwen3 embedding: Advancing text embedding and reranking through foundation models.
\newblock \emph{arXiv preprint arXiv:2506.05176}, 2025.

\bibitem[Zhou et~al.(2023)Zhou, Lu, Mishra, Brahma, Basu, Luan, Zhou, and Hou]{zhou2023instruction}
Jeffrey Zhou, Tianjian Lu, Swaroop Mishra, Siddhartha Brahma, Sujoy Basu, Yi~Luan, Denny Zhou, and Le~Hou.
\newblock Instruction-following evaluation for large language models.
\newblock \emph{arXiv preprint arXiv:2311.07911}, 2023.

\bibitem[Zhou et~al.(2024)Zhou, Zheng, Chen, Zheng, Shang, Zhang, Meng, and Shen]{Zhou2024BeyondCR}
Jianqun Zhou, Yuanlei Zheng, Wei Chen, Qianqian Zheng, Zeyuan Shang, Wei Zhang, Rui Meng, and Xiaoyu Shen.
\newblock Beyond content relevance: Evaluating instruction following in retrieval models.
\newblock \emph{ArXiv}, abs/2410.23841, 2024.
\newblock URL \url{https://api.semanticscholar.org/CorpusID:273707185}.

\end{thebibliography}
\bibliographystyle{iclr2026_conference}

\appendix

\newpage

\section{Relationship to Order-K Voronoi Regions}
We also provide an explanation for how our results compare to \citet{clarkson1988applications} which put bounds on the number of regions in the order-$k$ Voronoi graph. The order-$k$ Voronoi graph is defined as the set of points having a particular set of $n$ points in $S$ as its $n$ nearest neighbors. This maps nicely to retrieval, as each order-$k$ region is equivalent to one retrieved set of top-$k$ results. Then the count of unique regions in the Voronoi graph is the total number of combinations that could be returned for those points. However, creating an empirical order-k Voronoi graph is computationally infeasible for $d$ > 3, and theoretically it is hard to bound tightly. Thus we use a different approach for showing the limitations of embedding models.

\section{Hyperparameter and Compute Details}
\paragraph{Inference} We use the default length settings for evaluating models using the MTEB framework \citep{enevoldsen2025mmteb}. As our dataset has relatively short documents (around 100 tokens), this does not cause an issue. 

\paragraph{Training} For training on the \datasetname{} training and test set we use the SentenceTransformers library \citep{reimers-2019-sentence-bert} using the MultipleNegativesRankingLoss. We use a full dataset batch size and employ the no duplicates sampler to ensure that no in-batch negatives are duplicates of the positive docs. We use a learning rate of 5e-5. We train for 5 epochs and limit the training set slightly to the size of the test set (from 2.5k to 2k examples, matching test).

\paragraph{Compute} Inference and training for \datasetname{} is done with A100 GPUs on Google Colab Pro. The free embedding experiments are done mainly on H100 GPUs and TPU v5's for larger size $N$ to accommodate higher VRAM for full-dataset batch vector optimization.

\section{Sigmoid Loss Function for Free Embeddings}
\label{app:sigmoid}
In concurrent work by \citet{bangachev2025global}, they show that for vision-language embedding models like CLIP (that more commonly use sigmoid loss functions) that the free-embedding experiments can be solved in fewer dimensions than in our setting (assuming  no margin). As our results found the best performance with InfoNCE, which attempts to create the widest possible margin, this indicates that there are additional questions to resolve around learnability. We welcome further insight into this question both theoretically and empirically, as there exists widely disparate practices between the vision-language community (where sigmoid loss functions are often SOTA) and the text-only community (where sigmoid loss functions are almost never used due to worse performance). 

This sigmoid learning is also closely related to other work, such as \citet{grivas2024taming} and generally on the topic of work such as \citep{badreddine2025breaking,paul2021multi}

\section{Proof using Sign-Rank}
In the initial version of this paper, we provided a theoretical bound without any margin requirement based on the qrel matrices \emph{sign rank}. Although the proof is correct, the sign rank of the $\binom{n}{k}$ matrix has been established in previous work \citep{alongeometrical} and only depends on k. We include the proof connecting notions relevant for retrieval with the classic notion of sign-rank, however we emphasize that this will provide a weaker requirement on dimension as it assumes no margin.

\subsection{Formalization}

We consider a set of $m$ queries and $n$ documents with a ground-truth relevance matrix $A \in \{0, 1\}^{m \times n}$, where $A_{ij} = 1$ if and only if document $j$ is relevant to query $i$.\footnote{The matrix $A$ is often called the "qrels" (query relevance judgments) matrix in information retrieval.} Vector embedding models map each query to a vector $u_i \in \mathbb{R}^d$ and each document to a vector $v_j \in \mathbb{R}^d$. Relevance is modeled by the dot product $u_i^T v_j$, with the goal that relevant documents should score higher than irrelevant ones.

Concatenating the vectors for queries in a matrix $U \in \RR^{d \times m}$ and those for documents in a matrix $V \in \RR^{d \times n}$, these dot products are the entries of the score matrix $B = U^TV$. The smallest embedding dimension $d$ that can realize a given score matrix is, by definition, the rank of $B$. Therefore, our goal is equivalent to finding the minimum rank of a score matrix $B$ that correctly orders documents according to the relevance specified in $A$, which we formalize in the following definition.

\begin{definition}
    Given a matrix $A \in \RR^{m \times n}$, the \defterm{row-wise order-preserving rank of $A$} is the smallest integer $d$ such that there exists a rank-$d$ matrix $B$ that preserves the relative order of entries in each row of $A$. We denote this as
    \[\rank_\text{rop}A = \min \{\rank B \mid B \in \mathbb R^{m \times n}, \text{ such that for all } i, j, k, \text{ if } A_{ij} > A_{ik} \text{ then } B_{ij} > B_{ik}\}.\]
\end{definition}

In other words, if $A$ is a binary ground-truth relevance matrix, $\rank_\text{rop} A$ is the minimum dimension necessary for any vector embedding model to return relevant documents before irrelevant ones for all queries.  Alternatively, we might require that the scores of relevant documents can be cleanly separated from those of irrelevant ones by a threshold.

\begin{definition}
    Given a binary matrix $A\in \{0,1\}^{m \times n}$:
    \begin{itemize}[leftmargin=0.9em]
        \item The \defterm{row-wise thresholdable rank of $A$} ($\rank_\text{rt} A$) is the minimum rank of a matrix $B$ for which there exist row-specific thresholds $\{\tau_i\}_{i=1}^m$ such that for all $i,j$, $B_{ij} > \tau_i$ if $A_{ij} = 1$ and $B_{ij} < \tau_i$ if $A_{ij} = 0$.
        \item The \defterm{globally thresholdable rank of $A$} ($\rank_\text{gt} A$) is the minimum rank of a matrix $B$ for which there exists a single threshold $\tau$ such that for all $i,j$, $B_{ij} > \tau$ if $A_{ij} = 1$ and $B_{ij} < \tau$ if $A_{ij} = 0$.
    \end{itemize}
\end{definition}

\begin{remark}
\label{rem: threshold separation}
This two-sided separation condition may be seen as slightly stronger than requiring $B_{ij} > \tau_i$ if and only if $A_{ij} = 1$, however since there are only finitely many elements of $B_{ij}$ we could always perturb the latter threshold by a sufficient number such that the two-sided condition holds.\footnote{Without loss of generality, we may assume the thresholds in the above definitions are not equal to any elements of $B$ since we could increase the threshold of $\tau$ by a sufficiently small $\epsilon$ to preserve the inequality.}
\end{remark}

\subsection{Theoretical Bounds}

For binary matrices, row-wise ordering/thresholding are equivalent notions of representation capacity.

\begin{proposition}
\label{prop: rop equal to rt}
For a binary matrix $A\in\{0,1\}^{m \times n}$, we have that $\rank_\text{rop} A = \rank_\text{rt} A$.
\end{proposition}
\begin{proof}
($\le$) Suppose $B$ and $\tau$ satisfy the row-wise thresholdable rank condition. Since $A$ is a binary matrix $A_{ij} > A_{ik}$ implies $A_{ij} = 1$ and $A_{ik} = 0$, thus $B_{ij} > \tau_i > B_{ik}$, and hence $B$ also satisfies the row-wise order-preserving condition.

($\ge$) Let $B$ satisfy the row-wise order-preserving condition, so $A_{ij} > A_{ik}$ implies $B_{ij} > B_{ik}$. For each row $i$, let $U_i = \{B_{ij} \mid A_{ij}=1\}$ and $L_i = \{B_{ij} \mid A_{ij} = 0\}$. The row-wise order-preserving condition implies that every element of $U_i$ is greater than every element of $L_i$. We can therefore always find a threshold $\tau_i$ separating them (\eg $\tau_i = (\max L_i + \min U_i) / 2$ if both are non-empty, trivial otherwise). Thus $B$ is also row-wise thresholdable to $A$.
\end{proof}

The notions we have described so far are closely related to the sign rank of a matrix, which we use in the rest of the paper to establish our main bounds.

\begin{definition}[Sign Rank] The sign rank of a matrix $M \in \{-1,1\}^{m\times n}$ is the smallest integer $d$ such that there exists a rank $d$ matrix $B \in \mathbb R^{m \times n}$ whose entries have the same sign as those of $M$, i.e.
\[\rank_\pm{M} = \min \{\rank B \mid B \in \mathbb R^{m \times n} \text{ such that for all } i, j \text{ we have } \sign B_{ij} = M_{ij}\}.\]
\end{definition}

In what follows, we use $\mathbf 1_n$ to denote the $n$-dimensional vector of ones, and $\mathbf 1_{m \times n}$ to denote an $m \times n$ matrix of ones.
\begin{proposition}
\label{prop: ordering relationship of ranks}
Let $A\in \{0,1\}^{m \times n}$ be a binary matrix. Then $2A - \mathbf 1_{m \times n} \in \{-1,1\}^{m \times n}$ and
\[\rank_\pm (2A - \mathbf 1_{m \times n}) - 1\le \rank_\text{rop} A = \rank_\text{rt} A \le \rank_\text{gt} A \le \rank_\pm (2A - \mathbf 1_{m \times n})\]
\end{proposition}
\begin{proof}
N.b. the equality was already shown in \Cref{prop: rop equal to rt}. We prove each inequality separately.

\textbf{1. $\rank_\text{rt} A \le \rank_\text{gt} A$:} True by definition, since any matrix satisfying the globally thresholdable condition trivially satisfies a row-wise thresholdable condition with the same threshold for each row.

\textbf{2. $\rank_\text{gt} A \le \rank_\pm(2A-\mathbf 1_{m \times n})$:} Let $B$ be any matrix whose entries have the same sign as $2A-\mathbf 1_{m \times n}$, then 
\[B_{ij} > 0 \iff 2A_{ij}-1> 0 \iff A_{ij} = 1.\]
Thus $B$ satisfies the globally thresholdable condition with a threshold of $0$.

\textbf{3. $\rank_\pm(2A-\mathbf 1 _{m \times n}) - 1 \le \rank_\text{rt}A$:}
Suppose $B$ satisfies the row-wise thresholdable condition with minimal rank, so $\rank_\text{rt} A = \rank B$ and there exists $\tau \in \RR^m$ such that $B_{ij} > \tau_i$ if $A_{ij} = 1$ and $B_{ij} < \tau_i$ if $A_{ij} = 0$. 
Then the entries of $B - \tau \mathbf 1_n^T$ have the same sign as $2A-\mathbf 1_{m \times n}$, since $(B - \tau \mathbf 1_n^T)_{ij} = B_{ij} - \tau _i$ and
\begin{align}
B_{ij} - \tau_i > 0 &\iff A_{ij} = 1 \iff 2A_{ij} - 1 > 0,\text{ and}\\
B_{ij} - \tau_i < 0 &\iff A_{ij} = 0 \iff 2A_{ij} - 1 < 0.
\end{align}
Thus $\rank_\pm(2A - \mathbf 1_{m \times n}) \le \rank(B - \tau \mathbf 1_n^T) \le \rank(B) + \rank(\tau \mathbf 1_n^T) = \rank_\text{rt} A + 1$.

Combining these gives the desired chain of inequalities.
\end{proof}

\subsection{Consequences}
In the context of a vector embedding model, this provides a lower and upper bound on the dimension of vectors required to exactly capture a given set of retrieval objectives, in the sense of row-wise ordering, row-wise thresholding, or global thresholding. In particular, given some binary relevance matrix $A\in \{0,1\}^{m \times n}$, we need at least $\rank_\pm (2A - \mathbf 1_{m \times n}) - 1$ dimensions to capture the relationships in $A$ exactly, and can always accomplish this in at most $\rank_\pm(2A - \mathbf 1_{m \times n})$ dimensions.

The cyclotomic polynomial construction presented in \citet{alongeometrical} implies that any qrel matrix has sign-rank at most $2k$, where $k$ is the largest number of documents for a particular query. The construction results in unnormalized vectors, however this can be easily adapted to normalized vectors by using one additional dimension. In agreement with \Cref{thm:margin-lb}, this construction requires infinite precision in general, and is thus not feasible in practice.

\subsection{Correlation with MTEB}

\begin{wrapfigure}[7]{r}{0.35\textwidth}
    \vspace{-4em}
    \includegraphics[width=0.35\textwidth]{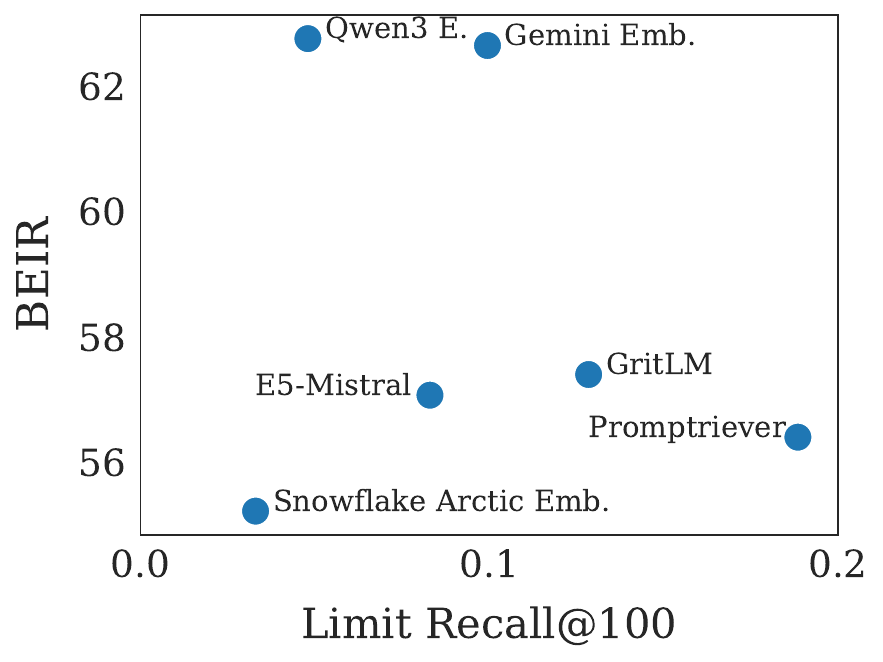}
    \vspace{-2em}
    \caption{\small{No obvious correlation between BEIR vs LIMIT.}}
    \label{fig:mteb_vs_limit}
\end{wrapfigure}

BEIR (used in MTEB v1) \citep{thakur2021beir,muennighoff2022mteb} has frequently been cited as something that embedding models have overfit to \citep{weller2025rank1,thakur2025freshstack}. We compare performance on \datasetname{} to BEIR in Figure~\ref{fig:mteb_vs_limit}. We see that performance is generally not correlated and that smaller models (like Arctic Embed) do worse on both, likely due to embedding dimension and pre-trained model knowledge.

\section{LLM Usage}
LLMs were not used for any paper writing, only for coding help and title brainstorming.

\section{Metrics Measuring Qrel Graph Density}
\label{app:measures}
We show two metrics that treat the qrel matrix as a graph and show that \datasetname{} has unique properties compared to standard IR datasets (Table~\ref{tab:metrics}). We call these metrics Graph Density and Average Query Strength and describe them below. 

\paragraph{Graph Density}
We use the qrel matrix to construct the graph, where nodes are documents and an edge exists between two documents if they are both relevant to at least one common query.

For a given graph $G=(V, E)$ with $V$ being the set of nodes and $E$ being the set of edges, the graph density is defined as the ratio of the number of edges in the graph to the maximum possible number of edges. For an undirected graph, the maximum possible number of edges is $\frac{|V|(|V|-1)}{2}$. Thus, the density $\rho$ is calculated as:

$$
\rho = \frac{|E|}{\frac{|V|(|V|-1)}{2}} = \frac{2|E|}{|V|(|V|-1)}
$$

This metric indicates how connected the graph is; a density of 1 signifies a complete graph (all possible edges exist), while a density close to 0 indicates a sparse graph. For a qrel dataset, the

\paragraph{Average Query Strength}
In a query-query graph where nodes are queries and edges represent similarity between queries (e.g., Jaccard similarity of their relevant documents), the \textit{strength} of a query node $i$, denoted $s_i$, is defined as the sum of the weights of all edges incident to it. If $w_{ij}$ is the weight of the edge between query $i$ and query $j$, and $N(i)$ is the set of neighbors of query $i$, then the strength is:

$$
s_i = \sum_{j \in N(i)} w_{ij}
$$

The Average Query Strength $\bar{s}$ is the mean of these strengths across all query nodes in the graph:

$$
\bar{s} = \frac{1}{|V_Q|} \sum_{i \in V_Q} s_i
$$

where $V_Q$ is the set of all query nodes in the graph. This metric provides an overall measure of how strongly connected queries are to each other on average within the dataset, based on their shared relevant documents.

\paragraph{Comparisons to other datasets}
We compare with standard IR Datasets such as NQ \citep{kwiatkowski2019natural}, HotpotQA \citep{yang2018hotpotqa}, and SciFact \citep{wadden2020fact}. We also show an instruction-following dataset, FollowIR Core17 \citep{weller2024followir}. For all datasets, we use the test set only. The results in Table~\ref{tab:metrics} show that \datasetname{} has significantly higher values for both of these metrics (i.e. 28 for query similarity compared to 0.6 or lower for the others).

\begin{table}[h!]
    \centering
    \captionsetup{justification=raggedright, singlelinecheck=false}
    \caption{Metrics measuring the density of the qrel matrix. We see that \datasetname{} is significantly higher than other datasets, but that the closest are instruction-following datasets such as Core17 from FollowIR. Our empirical ablations suggest (although cannot definitively prove) that datasets with higher values here will be harder for retrieval models to represent.}
    \label{tab:metrics}
    \begin{tabular}{lrr}
        \toprule
        \textbf{Dataset Name} & \textbf{Graph Density} & \textbf{Average Query Strength} \\
        \midrule
        NQ & 0 & 0 \\
        HotPotQA & 0.000037 & 0.1104 \\
        SciFact & 0.001449 & 0.4222 \\
        FollowIR Core17 & 0.025641 & 0.5912 \\
        \datasetname{} & 0.085481 & 28.4653 \\
        \bottomrule
    \end{tabular}
\end{table}

\section{Table Forms of Figures}
In this section we show the table form of various figures. For Figure~\ref{fig:real} it is Table~\ref{tab:real}, Figure~\ref{fig:real_small} in Table~\ref{tab:real_small}, Figure~\ref{fig:critical_n} in Table~\ref{tab:critical_values}, and Figure~\ref{fig:fine_tuning} in Table~\ref{tab:fine_tuning}.

\begin{table}
\centering
\begin{tabular}{ll|rrr}
\toprule
Split & Dim & Recall@2 & Recall@10 & Recall@100 \\
\midrule
Test & 32 & 85.5 & 98.4 & 100.0 \\
Test & 64 & 90.4 & 98.7 & 100.0 \\
Test & 128 & 93.1 & 99.5 & 99.9 \\
Test & 256 & 94.2 & 99.7 & 100.0 \\
Test & 384 & 95.6 & 99.6 & 100.0 \\
Test & 512 & 94.0 & 99.5 & 99.9 \\
Test & 768 & 96.1 & 99.8 & 100.0 \\
Test & 1024 & 96.5 & 99.8 & 100.0 \\
\midrule
Train & 32 & 0.0 & 0.0 & 0.0 \\
Train & 64 & 0.1 & 0.3 & 2.2 \\
Train & 128 & 0.2 & 0.7 & 3.1 \\
Train & 256 & 0.0 & 0.0 & 0.4 \\
Train & 384 & 1.1 & 2.7 & 8.3 \\
Train & 512 & 0.7 & 2.3 & 9.8 \\
Train & 768 & 0.7 & 2.4 & 9.9 \\
Train & 1024 & 1.0 & 2.8 & 11.2 \\
\bottomrule
\end{tabular}
\caption{Fine-tuning results in table form. See Figure~\ref{fig:fine_tuning} for the comparable plot.}
\label{tab:fine_tuning}
\end{table}

\begin{table}
\scriptsize
\centering
\begin{tabular}{ll|rrr}
\toprule
Model & Dim & Recall@2 & Recall@10 & Recall@20 \\
\midrule
BM25 & default & 97.8 & 100.0 & 100.0 \\
E5-Mistral 7B & 32 & 7.9 & 32.6 & 56.2 \\
E5-Mistral 7B & 64 & 10.2 & 37.0 & 60.3 \\
E5-Mistral 7B & 128 & 14.5 & 41.9 & 65.9 \\
E5-Mistral 7B & 256 & 15.3 & 45.9 & 69.7 \\
E5-Mistral 7B & 512 & 22.2 & 54.7 & 74.8 \\
E5-Mistral 7B & 768 & 21.6 & 57.5 & 79.2 \\
E5-Mistral 7B & 1024 & 24.5 & 60.5 & 80.0 \\
E5-Mistral 7B & 2048 & 28.9 & 66.3 & 83.2 \\
E5-Mistral 7B & 3072 & 29.9 & 67.8 & 85.3 \\
E5-Mistral 7B & 4096 & 29.5 & 68.1 & 85.2 \\
GTE-ModernColBERT & default & 83.5 & 97.6 & 99.1 \\
GritLM 7B & 32 & 7.8 & 33.5 & 56.3 \\
GritLM 7B & 64 & 9.4 & 35.9 & 59.6 \\
GritLM 7B & 128 & 14.2 & 42.7 & 64.9 \\
GritLM 7B & 256 & 17.3 & 46.2 & 68.3 \\
GritLM 7B & 512 & 21.8 & 55.6 & 76.7 \\
GritLM 7B & 768 & 23.8 & 58.1 & 80.1 \\
GritLM 7B & 1024 & 26.2 & 61.4 & 80.1 \\
GritLM 7B & 2048 & 33.0 & 69.1 & 86.2 \\
GritLM 7B & 3072 & 36.3 & 72.9 & 89.9 \\
GritLM 7B & 4096 & 38.4 & 75.4 & 90.5 \\
Promptriever Llama3 8B & 32 & 6.1 & 31.4 & 56.0 \\
Promptriever Llama3 8B & 64 & 8.9 & 35.8 & 62.3 \\
Promptriever Llama3 8B & 128 & 13.7 & 44.5 & 67.6 \\
Promptriever Llama3 8B & 256 & 18.5 & 52.1 & 74.1 \\
Promptriever Llama3 8B & 512 & 27.0 & 61.8 & 81.7 \\
Promptriever Llama3 8B & 768 & 35.5 & 69.0 & 84.7 \\
Promptriever Llama3 8B & 1024 & 38.0 & 73.5 & 89.1 \\
Promptriever Llama3 8B & 2048 & 46.2 & 83.6 & 94.2 \\
Promptriever Llama3 8B & 3072 & 49.2 & 87.3 & 96.6 \\
Promptriever Llama3 8B & 4096 & 54.3 & 90.0 & 97.7 \\
Qwen3 Embed & 32 & 8.3 & 30.6 & 53.9 \\
Qwen3 Embed & 64 & 9.4 & 35.5 & 57.6 \\
Qwen3 Embed & 128 & 11.6 & 38.3 & 60.8 \\
Qwen3 Embed & 256 & 14.3 & 41.6 & 63.8 \\
Qwen3 Embed & 512 & 16.1 & 43.7 & 66.0 \\
Qwen3 Embed & 768 & 17.2 & 45.3 & 69.3 \\
Qwen3 Embed & 1024 & 17.8 & 48.7 & 70.3 \\
Qwen3 Embed & 2048 & 19.5 & 51.5 & 72.4 \\
Qwen3 Embed & 3072 & 19.3 & 52.8 & 73.3 \\
Qwen3 Embed & 4096 & 19.0 & 52.3 & 73.8 \\
Gemini Embed & 2 & 4.2 & 23.0 & 45.5 \\
Gemini Embed & 4 & 4.2 & 21.9 & 46.0 \\
Gemini Embed & 8 & 4.9 & 23.2 & 47.0 \\
Gemini Embed & 16 & 5.2 & 24.7 & 47.5 \\
Gemini Embed & 32 & 6.3 & 25.2 & 50.6 \\
Gemini Embed & 64 & 6.9 & 30.6 & 55.0 \\
Gemini Embed & 128 & 7.7 & 37.0 & 62.9 \\
Gemini Embed & 256 & 14.6 & 46.9 & 69.7 \\
Gemini Embed & 512 & 23.3 & 58.4 & 77.9 \\
Gemini Embed & 768 & 28.8 & 67.5 & 84.5 \\
Gemini Embed & 1024 & 31.8 & 69.9 & 86.1 \\
Gemini Embed & 2048 & 31.9 & 70.3 & 87.1 \\
Gemini Embed & 3072 & 33.7 & 72.4 & 87.9 \\
Snowflake Arctic L & 32 & 8.3 & 30.3 & 53.8 \\
Snowflake Arctic L & 64 & 9.0 & 35.4 & 58.5 \\
Snowflake Arctic L & 128 & 12.7 & 41.3 & 65.1 \\
Snowflake Arctic L & 256 & 16.0 & 48.2 & 72.6 \\
Snowflake Arctic L & 512 & 16.7 & 51.3 & 74.1 \\
Snowflake Arctic L & 768 & 17.9 & 53.5 & 74.6 \\
Snowflake Arctic L & 1024 & 19.4 & 54.9 & 76.0 \\
Snowflake Arctic L & 2048 & 19.4 & 54.9 & 76.0 \\
Snowflake Arctic L & 3072 & 19.4 & 54.9 & 76.0 \\
Snowflake Arctic L & 4096 & 19.4 & 54.9 & 76.0 \\
\bottomrule
\end{tabular}
\caption{Results for the \datasetname{} small version. See comparable Figure~\ref{fig:real_small}.}
\label{tab:real_small}
\end{table}

\begin{table}
\scriptsize
\centering
\begin{tabular}{ll|rrr}
\toprule
Model & Dim & Recall@2 & Recall@10 & Recall@100 \\
\midrule
E5-Mistral 7B & 32 & 0.0 & 0.0 & 0.5 \\
E5-Mistral 7B & 64 & 0.0 & 0.1 & 0.4 \\
E5-Mistral 7B & 128 & 0.1 & 0.3 & 1.0 \\
E5-Mistral 7B & 256 & 0.4 & 0.9 & 1.9 \\
E5-Mistral 7B & 512 & 0.7 & 1.3 & 3.8 \\
E5-Mistral 7B & 768 & 0.9 & 1.7 & 4.3 \\
E5-Mistral 7B & 1024 & 0.9 & 1.8 & 5.9 \\
E5-Mistral 7B & 2048 & 1.0 & 1.9 & 6.8 \\
E5-Mistral 7B & 3072 & 1.3 & 2.0 & 7.7 \\
E5-Mistral 7B & 4096 & 1.3 & 2.2 & 8.3 \\
Snowflake Arctic L & 32 & 0.0 & 0.1 & 0.6 \\
Snowflake Arctic L & 64 & 0.2 & 0.4 & 1.7 \\
Snowflake Arctic L & 128 & 0.1 & 0.3 & 1.8 \\
Snowflake Arctic L & 256 & 0.2 & 0.8 & 2.5 \\
Snowflake Arctic L & 512 & 0.3 & 1.0 & 2.5 \\
Snowflake Arctic L & 768 & 0.4 & 1.1 & 3.1 \\
Snowflake Arctic L & 1024 & 0.4 & 0.8 & 3.3 \\
Snowflake Arctic L & 2048 & 0.4 & 0.8 & 3.3 \\
Snowflake Arctic L & 3072 & 0.4 & 0.8 & 3.3 \\
Snowflake Arctic L & 4096 & 0.4 & 0.8 & 3.3 \\
GritLM 7B & 32 & 0.0 & 0.0 & 0.8 \\
GritLM 7B & 64 & 0.0 & 0.1 & 0.3 \\
GritLM 7B & 128 & 0.1 & 0.3 & 1.3 \\
GritLM 7B & 256 & 0.1 & 0.4 & 2.8 \\
GritLM 7B & 512 & 0.6 & 1.8 & 6.5 \\
GritLM 7B & 768 & 1.5 & 3.1 & 8.7 \\
GritLM 7B & 1024 & 1.8 & 3.5 & 10.6 \\
GritLM 7B & 2048 & 2.3 & 4.3 & 11.8 \\
GritLM 7B & 3072 & 2.0 & 4.3 & 12.9 \\
GritLM 7B & 4096 & 2.4 & 4.1 & 12.9 \\
Promptriever Llama3 8B & 32 & 0.0 & 0.0 & 0.1 \\
Promptriever Llama3 8B & 64 & 0.0 & 0.0 & 0.3 \\
Promptriever Llama3 8B & 128 & 0.0 & 0.1 & 0.6 \\
Promptriever Llama3 8B & 256 & 0.2 & 0.4 & 1.8 \\
Promptriever Llama3 8B & 512 & 0.6 & 1.4 & 5.4 \\
Promptriever Llama3 8B & 768 & 1.3 & 3.1 & 8.7 \\
Promptriever Llama3 8B & 1024 & 2.1 & 4.4 & 12.8 \\
Promptriever Llama3 8B & 2048 & 3.2 & 6.5 & 18.1 \\
Promptriever Llama3 8B & 3072 & 2.9 & 6.3 & 17.8 \\
Promptriever Llama3 8B & 4096 & 3.0 & 6.8 & 18.9 \\
Qwen3 Embed & 32 & 0.0 & 0.1 & 1.1 \\
Qwen3 Embed & 64 & 0.0 & 0.2 & 1.0 \\
Qwen3 Embed & 128 & 0.3 & 0.4 & 1.8 \\
Qwen3 Embed & 256 & 0.4 & 0.8 & 3.2 \\
Qwen3 Embed & 512 & 0.6 & 1.3 & 3.3 \\
Qwen3 Embed & 768 & 0.7 & 1.5 & 3.8 \\
Qwen3 Embed & 1024 & 0.7 & 1.6 & 4.6 \\
Qwen3 Embed & 2048 & 0.9 & 1.7 & 4.7 \\
Qwen3 Embed & 3072 & 0.8 & 1.6 & 4.8 \\
Qwen3 Embed & 4096 & 0.8 & 1.8 & 4.8 \\
Gemini Embed & 2 & 0.0 & 0.0 & 0.1 \\
Gemini Embed & 4 & 0.0 & 0.0 & 0.0 \\
Gemini Embed & 8 & 0.0 & 0.0 & 0.0 \\
Gemini Embed & 16 & 0.0 & 0.0 & 0.0 \\
Gemini Embed & 32 & 0.0 & 0.0 & 0.0 \\
Gemini Embed & 64 & 0.0 & 0.0 & 0.3 \\
Gemini Embed & 128 & 0.0 & 0.1 & 0.3 \\
Gemini Embed & 256 & 0.0 & 0.1 & 1.2 \\
Gemini Embed & 512 & 0.2 & 1.1 & 3.6 \\
Gemini Embed & 768 & 0.9 & 2.5 & 7.6 \\
Gemini Embed & 1024 & 1.3 & 2.7 & 8.1 \\
Gemini Embed & 2048 & 1.5 & 3.1 & 8.5 \\
Gemini Embed & 3072 & 1.6 & 3.5 & 10.0 \\
GTE-ModernColBERT & default & 23.1 & 34.6 & 54.8 \\
BM25 & default & 85.7 & 90.4 & 93.6 \\
\bottomrule
\end{tabular}
\caption{Results on \datasetname{}. See comparable Figure~\ref{fig:real}.}
\label{tab:real}
\end{table}

\begin{table}
\centering
\scriptsize
\begin{tabular}{cc}
    \toprule
    $d$ & Critical-$n$ \\
    \midrule
    4 & 10 \\
    5 & 14 \\
    6 & 19 \\
    7 & 24 \\
    8 & 28 \\
    9 & 32 \\
    10 & 36 \\
    11 & 42 \\
    12 & 47 \\
    13 & 54 \\
    14 & 62 \\
    15 & 70 \\
    16 & 79 \\
    17 & 89 \\
    18 & 99 \\
    19 & 109 \\
    20 & 120 \\
    21 & 132 \\
    22 & 144 \\
    23 & 157 \\
    24 & 170 \\
    25 & 184 \\
    26 & 198 \\
    27 & 213 \\
    28 & 229 \\
    29 & 245 \\
    30 & 261 \\
    31 & 278 \\
    32 & 296 \\
    33 & 314 \\
    34 & 333 \\
    35 & 352 \\
    36 & 372 \\
    37 & 392 \\
    38 & 413 \\
    39 & 434 \\
    40 & 460 \\
    41 & 484 \\
    42 & 505 \\
    43 & 545 \\
    44 & 605 \\
    45 & 626 \\
    \bottomrule
\end{tabular}
\caption{Critical Values of n for different d values in the Free Embedding optimization experiments. See Figure~\ref{fig:critical_n} for the corresponding figure.\label{tab:critical_values}}
\end{table}

\begin{table}
\centering
\begin{tabular}{lrr}
\toprule
Model & BEIR & LIMIT R@100 \\
\midrule
Snowflake Arctic & 55.22 & 3.3 \\
Promptriever & 56.40 & 18.9 \\
E5-Mistral & 57.07 & 8.3 \\
GritLM & 57.40 & 12.9 \\
Gemini Embed & 62.65 & 10.0 \\
Qwen3 Embed & 62.76 & 4.8 \\
\bottomrule
\end{tabular}
\caption{BEIR vs \datasetname{} results. See Figure~\ref{fig:mteb_vs_limit} for the comparable plot.}
\end{table}

\end{document}